\documentclass[12pt]{article}
\usepackage[margin=1in]{geometry} 
\usepackage[onehalfspacing]{setspace} 
\usepackage{graphicx}

\usepackage{microtype}
\usepackage{lmodern}
\usepackage[T1]{fontenc}
\usepackage[bottom]{footmisc}
\usepackage{lscape}
\usepackage{rotating}
\usepackage{subcaption}
\usepackage[capposition=top]{floatrow}

\usepackage{natbib}
\usepackage[usenames, dvipsnames]{xcolor}
\usepackage[colorlinks=true]{hyperref}
\hypersetup{linkcolor=Blue, citecolor=Blue}

\usepackage{amsmath, amssymb, amsthm, thmtools, dsfont, bbm, enumitem, soul}
\allowdisplaybreaks 

\usepackage{titlesec}
\titlespacing{\paragraph}{0pt}{2.25ex plus 1ex minus .2ex}{0.4em}

\theoremstyle{definition} 
\newtheorem{definition}{Definition}

\newtheorem{specification}{Specification}
\newtheorem{interaction}{Interaction}
\newtheorem{remark}{Remark}
\declaretheorem[style=definition]{example}
\renewcommand\thmcontinues[1]{Continued}

\renewcommand\thmcontinues[1]{Continued}

\theoremstyle{plain} 
\newtheorem{assumption}{Assumption}

\newtheorem{lemma}{Lemma}
\newtheorem{proposition}{Proposition}

\renewcommand{\P}{\mathbb{P}} 
\newcommand{\R}{\mathbb{R}} 
\newcommand{\E}{\mathbb{E}} 
\newcommand{\1}{\mathds{1}} 
\newcommand{\LP}{\mathbb{L}} 
\newcommand{\var}{\textnormal{var}} 
\newcommand{\ind}{\perp\!\!\!\!\!\!\perp} 
\DeclareMathOperator*{\argmax}{\textnormal{argmax}} 
\renewcommand{\text}[1]{\textnormal{#1}} 

\newcommand{\B}{\mathcal{B}}
\newcommand{\G}{\mathcal{G}}

\renewcommand{\S}{\mathcal{S}}

\newcommand{\V}{\mathcal{V}}
\newcommand{\Y}{\mathcal{Y}}

\begin{document}

\title{Interpreting TSLS Estimators in Information Provision Experiments\thanks{We are grateful to Isaiah Andrews, Anna Mikusheva, and Frank Schilbach for their guidance and support; to Josh Angrist, David Autor, Dylan Balla-Elliott, Simon Jäger, Haruki Kono, and participants at the MIT Behavioral, Econometrics, and Third-Year Lunches for helpful comments and conversations; to Simon Jäger, Chris Roth, Nina Roussille, and Benjamin Schoefer for providing us with an excellent replication package; to Olivier Coibion, Yuriy Gorodnichenko, and Saten Kumar for making public an excellent replication package; to Chantal Pezold for providing us with preliminary experimental data. We acknowledge financial support from the George and Obie Shultz Fund, and the National Science Foundation Graduate Research Fellowship under Grant No. 1745302.}}
\author{Vod Vilfort\thanks{%
Department of Economics, MIT, Cambridge, MA 02142, vod@mit.edu} \and Whitney Zhang\thanks{%
Department of Economics, MIT, Cambridge, MA 02142, zhangww@mit.edu}}
\date{June 21, 2024}
\maketitle

\begin{abstract}
To estimate the causal effects of beliefs on actions, researchers often run information provision experiments. We consider the causal interpretation of two-stage least squares (TSLS) estimators in these experiments. We characterize common TSLS estimators as weighted averages of causal effects, and interpret these weights under general belief updating conditions that nest parametric models from the literature. Our framework accommodates TSLS estimators for both passive and active control designs. Notably, we find that some passive control estimators allow for negative weights, which compromises their causal interpretation. We give practical guidance on such issues, and illustrate our results in two empirical applications. 
\end{abstract}

\pagebreak

\section{Introduction}
There has been a surge in research that uses the random provision of information to estimate the causal effects of beliefs on actions.\footnote{In macroeconomics, researchers have studied the effects of beliefs about inflation, GDP growth, and other macroeconomic indicators on firm and household decision making \citep{coibion2020inflation, coibion2021effect, coibion2022monetary, coibion2023forward, kumar2023effect}; in labor economics, researchers have studied the effects of beliefs about others' wages on one's own efforts and job search decisions \citep{cullen2022much, jager2022worker}, and how beliefs about labor market tightness affect support for unions \citep{pezold_labor_2023}; at the intersection of labor and public finance, researchers have studied the effect of beliefs about future government benefits on human capital investments \citep{deshpande2023lack}; yet others have studied the effects of beliefs about discrimination on policy preferences \citep{haaland2023racial, settele_how_2022}.} These information provision experiments provide a basis for testing the assumptions of economic models, differentiating across theoretical mechanisms, and informing economic policy \citep{haaland2023designing}. Given the importance of these goals, it is essential that the causal parameters of interest be carefully defined and accurately estimated. 

However, in practice, the causal parameters of interest are often informally defined, or are developed in stylized models that impose strong conditions on agents' beliefs, learning, and actions. In such models, it is unclear which of the implied restrictions drive the conclusions of an information provision experiment. Relatedly, it is ambiguous which of the many estimation strategies from the literature recover interpretable causal parameters. This paper addresses both of these concerns. 

An information provision experiment generally proceeds as follows. First, the experiment elicits \textit{features} (e.g., expectations) of agents' prior beliefs over a set of action-relevant states. Next, the experiment randomly assigns agents to different groups---either to a control group receiving no information, or to one of potentially multiple treatment groups receiving \textit{signals}, where a signal is a piece of relevant information. Then, the experiment elicits features of agents' posterior beliefs, and records the actions taken under those posterior beliefs. Finally, the experiment uses group assignment to instrument for beliefs in two-stage least squares (TSLS) regressions of actions on posterior features: TSLS specifications in \textit{passive control} experiments compare the control group to the treatment group(s), whereas specifications in \textit{active control} experiments compare across treatment groups.

In Section \ref{main:sec:setup}, we introduce an instrumental variables (IV) framework in which beliefs affect actions through features. We consider the \textit{partial effects} of features on actions, which we allow to vary across agents and feature values. In Section \ref{main:sec:TSLS}, we characterize TSLS estimators from the literature as weighted average partial effects (APEs) across these two margins. Notably, without further structure on belief updating to ensure IV monotonicity, the weights on agents' partial effects can be \textit{negative}, which compromises the causal interpretation of TSLS estimators \citep{imbens1994identification, blandhol2022tsls}. 

In Section \ref{main:sec:interpretation}, we propose conditions that ensure IV monotonicity. First, we propose \textit{signal monotonicity}, which formalizes the idea that agents should update their beliefs ``towards'' the signals. We motive signal monotonicity in a model of posterior formation where agents perceive the signals as realizations from probability distributions that satisfy the monotone likelihood ratio (MLR) property. This allows for a broad class of signal distributions and belief updating rules. In particular, signal monotonicity nests existing models that assume Gaussian distributions and Bayesian updating \citep{armantier2016price, cavallo2017inflation, armona2019home, cullen2022much, fuster2022expectations, balla2022determinants}. We additionally propose \textit{control-group stability} for passive control comparisons, and \textit{treatment-group neutrality} for active control comparisons; both formalize common experimental design considerations.

Under these conditions, TSLS estimators recover APEs that weight agents according to how much their beliefs respond to information provision. In particular, the weights for active control estimators emphasize agents with larger differences in their counterfactual posterior features across treatment groups (i.e., low versus high signal). Analogously, the weights for passive control estimators emphasize agents with larger differences in their counterfactual posterior features across treatment and control (i.e., signal versus no signal). However, some passive control estimators (i) up-weight agents with larger \textit{perception gaps} (i.e., the difference between prior features and signals); or (ii) are contaminated by linear combinations of the TSLS first-stage coefficients. The latter allows for negative weights, which we advise against. The former ensures non-negative weights, but leads to a different APE than the baseline case above. We interpret these differences and give practical recommendations in Section \ref{main:sec:recs}.

In Section \ref{main:sec:applications}, we take our results on passive control estimation to \cite{kumar2023effect} and \cite{jager2022worker}. In the \cite{kumar2023effect} application, we find that for three outcomes, the uncontaminated estimators produce coefficients that are one-third to two-thirds of the magnitude of those produced by the contaminated estimators; for one outcome, while the contaminated coefficients are statistically significant from the null, the uncontaminated coefficients are not. Together, we take this as evidence that negative weights meaningfully impact the estimates from the contaminated specifications. In the \cite{jager2022worker} application, for one outcome, the non-up-weighted estimator produces a coefficient that is 30-40\% larger in magnitude than those of the other specifications. Altogether, these empirical results highlight that the choice of weights can substantially impact the magnitude and significance of one's estimates.
 
\paragraph{Related Literature.}
Our formulation of the causal effects of beliefs as the partial effects of features aligns with existing empirical practice, and nests models considered in the literature where agents optimize their actions as a function of their beliefs \citep{cullen2022much, balla2023identifying, jager2022worker}. In allowing for heterogeneous partial effects, we are consistent with characterizations of IV estimators from other settings with continuous endogenous variables \citep{angrist2000interpretation, rambachan2021common, andrews2023causal}. To our knowledge, heterogeneity in the feature value margin has not received explicit attention in the information provision literature---we show in Section \ref{main:sec:recs} that accounting for this margin affects the interpretation of the policy counterfactuals captured by passive and active control experiments.

Our proposed conditions correspond to notions from the IV literature. For active control estimators, signal monotonicity and treatment-group neutrality correspond to the canonical IV monotonicity condition \citep{imbens1994identification, angrist2000interpretation}. In particular, agents' posterior features are larger when provided a high signal as opposed to a low signal. For passive control estimators, signal monotonicity and control-group stability correspond to ``weak'' IV monotonicity, which allows the direction of IV monotonicity to vary by covariates \citep{sloczynski2020should, blandhol2022tsls}. In particular, agents' posterior features are larger when provided a signal above their prior feature, and smaller when provided a signal below their prior feature. Finally, the contamination that we find in some passive control estimators is an instance of TSLS specifications failing to be ``monotonicity correct'' in their first-stage \citep{blandhol2022tsls}. 

\cite{haaland2023designing} survey applications of information provision experiments, and give guidance on experimental design, belief elicitation techniques, and other technical challenges. In contrast, we develop theory for the identification and estimation of causal effects. The closest to our paper in this regard is the independent and concurrent work of \cite{balla2023identifying}, who likewise studies the interpretation of TSLS estimators in information provision experiments. \cite{balla2023identifying} considers the partial effects of expectations, and targets an APE that places equal weights across agents. To identify this APE, \cite{balla2023identifying} restricts heterogeneity in the feature value margin, and appeals to the structure of (i) active control comparisons; and (ii) linear updating of expectations. In contrast, since we primarily seek to characterize existing specifications, we achieve identification under weaker conditions on agents' actions, in both passive and active control experiments, and for more general learning environments. Finally, while \cite{balla2023identifying} also interprets TSLS specifications from the literature, we discuss and characterize a more comprehensive set of specifications in a more general framework.

\section{Setup}\label{main:sec:setup}
Let $\Delta(\Omega)$ denote the set of probability distributions over states $\omega \in \Omega$. Agents' beliefs $B$ are contained in a subset of distributions $\B \subseteq \Delta(\Omega)$. Examples include (i) sets of distributions for which relevant moments exist; and (ii) families of parametric distributions considered in the literature. We use $\phi: \B \to \R$ to index features of interest. For instance, when $\Omega \subseteq \R$, examples include the mean $\mu(B) := \int \omega dB(\omega)$ and the variance $ \sigma^{2}(B) := \int (\omega - \mu(B))^{2} dB(\omega)$.

\subsection{Experiment}
An information provision experiment generally proceeds as follows. First, the experiment elicits features $\phi(B_{i0})$ of agents' prior beliefs $B_{i0}$. Next, the experiment randomly assigns agents to groups $g \in \G$. In passive control experiments, the groups are $\G = \{C, T\}$: Agents assigned to treatment receive signals $S_{i}^{T}$, whereas agents assigned to control receive no information $S_{i}^{C} := \varnothing$. In active control experiments, the groups are $\G = \{L, H\}$: Agents assigned to high treatment receive signals $S_{i}^{H}$, whereas agents assigned to low treatment receive signals $S_{i}^{L} < S_{i}^{H}$. 

Given group assignment $G_{i} \in \G$, agents form posterior beliefs $B_{i1} \equiv B_{i1}^{G_{i}}$. These posterior beliefs influence agents' actions/outcomes $Y_{i} \in \R$, which the experiment records; we assume there exist functions $Y_{i}^{g}(B)$ such that $Y_{i} \equiv Y_{i}^{G_{i}}(B_{i1})$. Finally, the experiment elicits posterior features $\phi(B_{i1})$. The goal is to estimate the causal effects of beliefs on actions. 

\begin{example}[name=Passive Control, label=main:lab:example:jager] \cite{jager2022worker} run a passive control experiment on a sample of workers from the German Socio-Economic Panel to investigate if workers have accurate beliefs about the wage distribution. 

In our notation, workers $i$ have beliefs $B$ about the wages $\omega$ at the best job they would find (i.e., their outside option) if forced to leave their current job. These beliefs influence workers' labor market behavior $Y_{i}$, such as their intended probability of looking for a new job or asking for a raise. The treatment group receives the average wage $S_{i}^{T}$ of workers with similar characteristics. \cite{jager2022worker} elicit workers' prior and posterior expectations $\mu(B_{i0}), \mu(B_{i1})$ of the wages at their outside options.
\end{example}

\begin{example}[name=Active Control, label=main:lab:example:roth]
Households' personal economic outlook and economic decisions depend on their macroeconomic expectations. To study these relationships, \cite{roth2020expectations} run an active control experiment on a sample of U.S. households. 

In our notation, households $i$ have beliefs $B$ about the probability $\omega$ of a recession. These beliefs influence households' actions/outcomes $Y_{i}$, such as their anticipated earnings growth and net stock purchases. For $g = L$, the experiment provides the predicted probability of a recession from a pessimistic forecaster. For $g = H$, the experiment provides the analogous information from an optimistic forecaster. Thus, $S_{i}^{L} < S_{i}^{H}$, with $S_{i}^{g}$ constant across $i$. \cite{roth2020expectations} elicit households' prior and posterior expectations $\mu(B_{i0}), \mu(B_{i1})$ of the probability of a recession.
\end{example}

\begin{remark}
We focus on experiments that provide quantitative information $S_{i}^{g} \in \R \cup \varnothing$, since that is the leading case in the literature. That said, there are settings where the information has qualitative components---Appendix \ref{main:sec:general.identification} covers these settings. 
\end{remark}

\begin{remark}
The results that follow apply to \textit{pairs} of comparison groups $\{g, \Tilde{g}\}$. Therefore, in experiments with more than two treatment groups, one can condition on pairs of comparison groups, apply our results, and then aggregate as necessary---see Appendix \ref{main:sec:general.estimation} for details. In considering pairs, we also avoid complications that arise with the causal interpretation of TSLS with multiple instruments \citep{mogstad2021causal}.
\end{remark}

\subsection{Instrumental Variables}\label{main:sec:IV}
The random assignment of agents to groups provides a basis for estimating causal effects. Formally, we assume that group assignment $G_{i}$ is a valid instrument. In what follows, $X_{i}$ is a vector of agent characteristics. 
\begin{assumption}[Valid Instrument]\label{main:ass:IV}
$G_{i}$ satisfies 
\begin{enumerate}[label=(\roman*)]
    \item independence: $G_{i} \ind (X_{i}, B_{i0}, \{S_{i}^{g}, B_{i1}^{g}, Y_{i}^{g}(B)\}_{g \in \G, B \in \B})$;
    \item exclusion: $Y_{i}^{g}(B) = Y_{i}(B)$, $\forall g, B$.
\end{enumerate}
\end{assumption}
Independence means that the experiment cannot condition group assignment $G_{i}$ on agent characteristics $X_{i}$, prior features $\phi(B_{i0})$, anticipated actions, and so on. However, it allows the content of the signals $S_{i}^{g}$ to depend on such variables \citep{jager2022worker, deshpande2023lack}. In any case, random assignment is sufficient for Assumption \ref{main:ass:IV}(i).

Exclusion means that group assignment only impacts actions through posterior beliefs. Therefore, one practical concern is that information provision may also affect actions through emotional responses \citep{haaland2023designing}. However, Assumption \ref{main:ass:IV}(ii) does accommodate emotional responses that solely affect posterior \textit{formation}---examples include belief-based utility or motivated reasoning \citep{brunnermeier2005optimal, epley2016mechanics}. 

\subsection{Feature Action Functions}
To formalize a tractable notion for the causal effects of beliefs  on actions, we assume that actions depend on beliefs through a single feature of interest $\phi$. In what follows, we suppose that the set of possible feature values $\phi(\B) := \{\phi(B): B \in \B\}$ is convex: For any two values in $\phi(\B)$, any third value between them can be rationalized by some $B \in \B$. This ensures that the causal effects defined in Section \ref{main:sec:TSLS} correspond to beliefs that the agents could actually hold.\footnote{Convexity holds whenever the set of possible beliefs $\B$ is sufficiently rich---see Appendix \ref{main:sec:primitive.conditions}.}

\begin{assumption}[Feature Action Functions]\label{main:ass:parametric}
$\phi(\B)$ is convex, and there exists a continuously differentiable function $Y_{i}^{\phi}(v)$ defined over it such that $Y_{i}(B) = Y_{i}^{\phi}(\phi(B))$.
\end{assumption}

Assumption \ref{main:ass:parametric} aligns with existing practice from the literature, which often frames the causal effects of beliefs on actions in terms of features. If $\B$ is parametrized by $\phi$, then actions depend on beliefs through $\phi$; examples include one-parameter exponential families \citep{lehmann2006theory}. If we broaden $\B$ to be the set of distributions with finite second moments, then another approach is to restrict preferences. For example, if we assume that agents are risk neutral in the sense that their optimal actions only depend on beliefs via first moments, then $Y_{i}(B) = Y_{i}^{\mu}(\mu(B))$. We formalize these arguments in Appendix \ref{main:sec:primitive.conditions}.

In some cases, Assumption \ref{main:ass:parametric} allows beliefs to affect actions through additional features. In particular, if $Y_{i}(B) = Y_{i}^{\phi, \eta}(\phi(B), \eta(B))$ and $\eta(B_{i1}^{g}) = \eta(B_{i1}^{\Tilde{g}})$ for $\G = \{g,\Tilde{g}\}$, then we can take $Y_{i}^{\phi}(\phi(B)) := Y_{i}^{\phi, \eta}(\phi(B), \eta(B_{i1}^{g}))$. This accommodates  models from the literature that predict $\sigma^{2}(B_{i1}^{H}) = \sigma^{2}(B_{i1}^{L})$ in active control experiments with $\phi = \mu$---see Appendix \ref{main:sec:gaussian.model}. An analogous argument applies for \textit{cross-learning}, which is when information about one state affects beliefs about another state \citep{haaland2023designing}. For example, information about inflation may influence expectations of both inflation \textit{and} economic growth \citep{coibion2023forward}. In such cases, it suffices to rule out \textit{differences} in cross-learning across $\{g,\Tilde{g}\}$.

One caveat is that, given a pair of comparison groups, we cannot allow multiple features to change at the same time. For example, we rule out comparisons where agents in one of the groups receive multiple signals with the intention of shifting multiple features \citep{cullen2022much, kumar2023effect, coibion2021effect}. Otherwise, there would be multiple endogenous variables; in such cases, TSLS estimators generally do not recover interpretable causal parameters \citep{bhuller20222sls}. 

\section{TSLS}\label{main:sec:TSLS}
We formulate the causal effects of beliefs on actions as the partial effects $\partial_{v}Y_{i}^{\phi}(v)$ of feature $\phi$.\footnote{In Appendix \ref{main:sec:alt.formulations}, we discuss alternative formulations, including one that accomodates discrete actions.} TSLS recovers APEs across feature values $v$ and agents $i$. In Sections \ref{main:sec:passive.specifications} and \ref{main:sec:active.specifications}, we derive these APEs for passive and active control specifications from the literature. In Section \ref{main:sec:takeaways.specifications}, we summarize key takeaways and motivate the need for assumptions on belief updating behavior. 

In what follows, $W_{i}$ is a vector that includes $1$ and potentially other variables that are independent of $G_{i}$. Let $\text{sign}(v) := \1\{v \geq 0\} - \1\{v \leq 0\}$ and $\psi_{i}^{g\tilde{g}} := \text{sign}(\phi(B_{i1}^{\tilde{g}}) - \phi(B_{i1}^{g}))$. Let $\lambda_{i}^{g\tilde{g}}(v)$ denote the density function of the uniform distribution on the values $v$ between $\phi(B_{i1}^{g})$ and $\phi(B_{i1}^{\tilde{g}})$. Propositions \ref{main:prop:passive.baseline} and \ref{main:prop:active.baseline} follow from TSLS algebra and the fundamental theorem of calculus, in the spirit of \citet[Theorem 1]{angrist2000interpretation}. See Appendix \ref{main:sec:proofs} for details.

\subsection{Passive Control TSLS}\label{main:sec:passive.specifications}

\begin{specification}\label{main:spec:passive}
Consider $I_{i}$, a sub-vector of $W_{i}$. Let $F_{i} = W_{i}'\pi_{0} + \1\{G_{i} = T\}I_{i}'\pi$ be the first-stage population linear regression of $\phi(B_{i1})$ on $W_{i}$ and $I_{i}\1\{G_{i} = T\}$. The passive control specification is  
\begin{align*}
\begin{split}
    \phi(B_{i1}) &= W_{i}'\pi_{0} + \1\{G_{i} = T\}I_{i}'\pi + \zeta_{i}, \\
    Y_{i} &= W_{i}'\gamma_{0} + \gamma^{CT} F_{i} + \upsilon_{i},
\end{split}
\end{align*}
where $\zeta_{i}$ and $\upsilon_{i}$ are population residuals. The coefficient of interest is $\gamma^{CT}$.
\end{specification}

\begin{proposition}\label{main:prop:passive.baseline}
Let $\G = \{C, T\}$ and suppose Assumptions \ref{main:ass:IV} and \ref{main:ass:parametric} are satisfied. If (i) $\E[W_{i}W_{i}']$ is full-rank and $\P(G_{i} = T) \in (0,1)$; and (ii) $\E[I_{i}\phi(B_{i1})|G_{i}=T]- \E[I_{i}\phi(B_{i1})|G_{i}=C] \neq 0$, then $\gamma^{CT}$ identifies
\begin{align}\label{main:eq:passive.baseline}
\begin{split}
    \beta^{CT} := \E[w_{i}^{CT}\Bar{\beta}_{i}^{CT}], \quad \quad w_{i}^{CT} &= \frac{|\phi(B_{i1}^{T}) - \phi(B_{i1}^{C})|\psi_{i}^{CT}I_{i}'\pi}{\E[|\phi(B_{i1}^{T}) - \phi(B_{i1}^{C})|\psi_{i}^{CT}I_{i}'\pi]}, \\[10pt]
    \Bar{\beta}_{i}^{CT} &:= \int \lambda_{i}^{CT}(v)\partial_{v}Y_{i}^{\phi}(v)dv.
\end{split}
\end{align}
\end{proposition}
Passive control specifications recover an APE, constructed as follows. First, for each agent $i$, the partial effect curve $ \partial_{v}Y_{i}^{\phi}(v)$ is aggregated across the range of values $v$ between $\phi(B_{i1}^{C})$ and $\phi(B_{i1}^{T})$. The density function $\lambda_{i}^{CT}(v)$ places uniform weight across this range of feature values, which generates $\Bar{\beta}_{i}^{CT}$. Then, $\Bar{\beta}_{i}^{CT}$ is averaged across agents with weights $w_{i}^{CT}$ that are proportional to (i) the signed difference $|\phi(B_{i1}^{T}) - \phi(B_{i1}^{C})|\psi_{i}^{CT}$ in $i$'s counterfactual posterior features across treatment and control; and (ii) the agent-specific linear combination $I_{i}'\pi$ of first-stage coefficients:
\begin{align*}
    w_{i}^{CT} \propto |\phi(B_{i1}^{T}) - \phi(B_{i1}^{C})|\psi_{i}^{CT}I_{i}'\pi.
\end{align*}
The weights $\lambda_{i}^{CT}(v)$ for the feature values are non-negative and integrate to one. The weights $w_{i}^{CT}$ for the agents integrate to one, but can be negative due to $\psi_{i}^{CT}$ and $I_{i}$. The literature considers a number of ``interactions'' $I_{i}$---note that $\pi$ cancels out when $I_{i}$ is a scalar.

\begin{interaction}\label{main:int:IPIV}
If $I_{i} = I_{i}^{sign} := \text{sign}(S_{i}^{T} - \phi(B_{i0}))$, then
\begin{align*}
    w_{i}^{CT} \propto |\phi(B_{i1}^{T}) - \phi(B_{i1}^{C})|\psi_{i}^{CT}\text{sign}(S_{i}^{T} - \phi(B_{i0})).
\end{align*}
This interaction is in the spirit of \cite{cantoni2019protests}. In the literature, $S_{i}^{T} - \phi(B_{i0})$ is known as the \textit{perception gap}.
\end{interaction}

\begin{interaction}\label{main:int:cullen}
If $I_{i} = I_{i}^{gap} := S_{i}^{T} - \phi(B_{i0})$, then
\begin{align*}
    w_{i}^{CT} \propto |\phi(B_{i1}^{T}) - \phi(B_{i1}^{C})|\psi_{i}^{CT}(S_{i}^{T} - \phi(B_{i0})).
\end{align*}
This interaction is in the spirit of \cite{cullen2022much}, \cite{galashin2020macroeconomic}, and \cite{balla2022determinants}. 
\end{interaction}

\begin{interaction}\label{main:int:jager}
If $I_{i} = I_{i}^{1, gap} := (1, S_{i}^{T} - \phi(B_{i0}))'$, then
\begin{align*}
    w_{i}^{CT} \propto |\phi(B_{i1}^{T}) - \phi(B_{i1}^{C})|\psi_{i}^{CT}(\pi_{1} + \pi_{2}(S_{i}^{T} - \phi(B_{i0}))).
\end{align*}
This interaction is in the spirit of \cite{jager2022worker}.
\end{interaction}

\begin{interaction}\label{main:int:kumar}
If $I_{i} = I_{i}^{1, prior} := (1, \phi(B_{i0}))'$, then
\begin{align*}
    w_{i}^{CT} \propto |\phi(B_{i1}^{T}) - \phi(B_{i1}^{C})|\psi_{i}^{CT}(\pi_{1}+ \pi_{2}\phi(B_{i0})).
\end{align*}
This interaction is in the spirit of \cite{kumar2023effect} and \citet{coibion2021effect, coibion2022monetary, coibion2023forward}. To the best of our knowledge, this interaction is only used when the signals $S_{i}^{T}$ are common across agents, which is allowed in our framework. This interaction is also analogous to an interaction from \cite{deshpande2023lack} with $I_{i} = (1, S_{i}^{T}, \phi(B_{i0}))'$. In this case, the signals must be heterogeneous, or else the first-stage coefficients cannot be identified.
\end{interaction}

\subsection{Active Control TSLS}\label{main:sec:active.specifications}

\begin{specification}\label{main:spec:active} $F_{i} = W_{i}'\pi_{0} + \pi \1\{G_{i} = H\}$ is the first-stage population linear regression of $\phi(B_{i1})$ on $W_{i}$ and $\1\{G_{i} = H\}$. The active control specification is
\begin{align*}
\begin{split}
    \phi(B_{i1}) &= W_{i}'\pi_{0} + \pi \1\{G_{i} = H\} + \zeta_{i}, \\
    Y_{i} &= W_{i}'\gamma_{0} + \gamma^{LH} F_{i} + \upsilon_{i},
\end{split}
\end{align*}
where $\zeta_{i}$ and $\upsilon_{i}$ are population residuals. The coefficient of interest is $\gamma^{LH}$.
\end{specification}

\begin{proposition}\label{main:prop:active.baseline}
Let $\G = \{L, H\}$ and suppose Assumptions \ref{main:ass:IV} and \ref{main:ass:parametric} are satisfied. If (i) $\E[W_{i}W_{i}']$ is full-rank and $\P(G_{i} = H) \in (0,1)$; and (ii) $\E[\phi(B_{i1})|G_{i} = H] - \E[\phi(B_{i1})|G_{i} = L] \neq 0$, then $\gamma^{LH}$ identifies 
\begin{align}\label{main:eq:active.baseline}
\begin{split}
    \beta^{LH} := \E[w_{i}^{LH}\Bar{\beta}_{i}^{LH}], \quad \quad w_{i}^{LH} &:= \frac{|\phi(B_{i1}^{H}) - \phi(B_{i1}^{L})|\psi_{i}^{LH}}{\E[|\phi(B_{i1}^{H}) - \phi(B_{i1}^{L})|\psi_{i}^{LH}]}, \\[10pt]
     \Bar{\beta}_{i}^{LH} &:= \int \lambda_{i}^{LH}(v)\partial_{v}Y_{i}^{\phi}(v)dv.
\end{split}
\end{align}
\end{proposition}
Active control specifications recover an APE that is constructed analogously to the passive control case. The weights $\lambda_{i}^{LH}(v)$ for the feature values are non-negative and integrate to one. The weights $w_{i}^{LH}$ for the agents integrate to one, but can be negative due to $\psi_{i}^{LH}$.

\subsection{Takeaways}\label{main:sec:takeaways.specifications}
Thus far, we assumed that (i) group assignment is a valid instrument for posterior beliefs; and (ii) actions depend on beliefs through a single feature of interest. Under these assumptions, we showed that TSLS specifications from the literature recover APEs $\beta^{CT}$ and $\beta^{LH}$ with weights $w_{i}^{CT}$ and $w_{i}^{LH}$ that can be negative. This negative weighting compromises the causal interpretation of $\beta^{CT}$ and $\beta^{LH}$. For instance, if some of the weights are negative, then it is possible for $\beta^{CT}$ and $\beta^{LH}$ to be negative even when all partial effects are positive \citep{blandhol2022tsls}. To address this concern, we now place structure on belief updating.

\section{Interpretation}\label{main:sec:interpretation}
The weights $w_{i}^{g\tilde{g}}$ are contaminated by $\psi_{i}^{g\tilde{g}} := \text{sign}(\phi(B_{i1}^{\tilde{g}}) - \phi(B_{i1}^{g}))$. This source of contamination suggests that we should restrict the ``direction'' of agents' responses to information provision. Section \ref{main:sec:signal.monotonicity} introduces a \textit{signal monotonicity} condition that makes this idea precise. To achieve IV monotonicity for passive control, Section \ref{main:sec:stability} additionally introduces a \textit{control-group stability} condition---Section \ref{main:sec:passive.parameters} then characterizes $w_{i}^{CT}$. To achieve IV monotonicity for active control, Section \ref{main:sec:neutrality} additionally introduces a \textit{treatment-group neutrality} condition---Section \ref{main:sec:active.parameters} then characterizes $w_{i}^{LH}$.

\subsection{Signal Monotonicity}\label{main:sec:signal.monotonicity}
For agent $i$ in group $g$, let $B_{i1}^{g}(\cdot|s)$ be the belief updating rule, which maps signals $s \in \R \cup \varnothing$ to beliefs $B \in \B$. Note that $B_{i1}^{g}(\cdot|s)$ implicitly depends on priors $B_{i0}$, and $B_{i1}^{g} \equiv B_{i1}^{g}(\cdot|S_{i}^{g})$ in the realized experiment.

\begin{assumption}[Signal Monotonicity]\label{main:ass:signal.monotonicity}
If $s' > s$, then $\phi(B_{i1}^{g}(\cdot|s')) \geq \phi(B_{i1}^{g}(\cdot|s))$.
\end{assumption}

Assumption \ref{main:ass:signal.monotonicity} restricts the direction of agents' responses to information provision; intuitively, agents should update ``towards'' the signal. This assumption is reasonable in many experiments. For example, in the setting of \cite{kumar2023effect}, firms in one treatment group are given the average forecasts of mean GDP from a panel of experts. It is reasonable to assume that larger average forecasts lead firms to form higher expectations of mean GDP. In another treatment group, firms are given the difference in forecasts between the most and least optimistic experts. It is reasonable to assume that larger dispersion in forecasts leads firms to be more uncertain in their beliefs.

However, Assumption \ref{main:ass:signal.monotonicity} can sometimes be invalid. For example, consider \cite{coibion2022monetary}, who study households' expectations of inflation. In one treatment group, they provide information about unemployment. Some households may believe that higher unemployment implies lower inflation, whereas others may believe the opposite. Therefore, signal monotonicity is unreasonable for this treatment group. 

As demonstrated, we can often intuit the validity of signal monotonicity. To complement this intuition, Appendix \ref{main:sec:posterior.formation} formalizes conditions under which Assumption \ref{main:ass:signal.monotonicity} is satisfied for the leading case of $\phi = \mu$. In summary, the main condition is that agents perceive the signals as realizations from distributions that satisfy a monotone likelihood ratio (MLR) property. In particular, agents assigned to treatment assume that the experiment tends to provide larger signals under larger realizations of the state.

The MLR property is satisfied for the set of Gaussian beliefs considered in the information provision literature---see Appendix \ref{main:sec:gaussian.model}. More generally, many exponential families satisfy the MLR property in their respective sufficient statistics \citep{casella2021statistical}. Given the MLR structure on agents' perceptions, we show in Appendix \ref{main:sec:posterior.formation} that Assumption \ref{main:ass:signal.monotonicity} is satisfied for $\phi = \mu$ under various belief updating rules, including (i) the Bayesian baseline considered in the information provision literature \citep{armantier2016price, cavallo2017inflation, armona2019home, cullen2022much, fuster2022expectations, balla2022determinants}; and (ii) systematic deviations from Bayesian updating considered in the behavioral economics literature \citep{gabaix2019behavioral, benjamin2019errors}. 

\subsection{Stability}\label{main:sec:stability}
We first consider passive control comparisons. Let $S_{i}^{\phi} := \phi(B_{i0})$ denote agent $i$'s prior feature.

\begin{definition}[Stability]
A passive control comparison is \textit{control-group stable} if
\begin{align*}
    \phi(B_{i1}^{T}(\cdot|S^{\phi}_{i})) = \phi(B_{i1}^{C}(\cdot|S^{C}_{i})).
\end{align*}
\end{definition}
If a passive control comparison is control-group stable, then the provision of $s = S_{i}^{\phi}$ in the treatment group is the same as not providing any information in the control group. Notice
\begin{align*}
    \psi_{i}^{CT} = \text{sign}([\phi(B_{i1}^{T}(\cdot|S^{T}_{i})) - \phi(B_{i1}^{T}(\cdot|S^{\phi}_{i}))] + [\phi(B_{i1}^{T}(\cdot|S^{\phi}_{i})) - \phi(B_{i1}^{C}(\cdot|S^{C}_{i}))]).
\end{align*}
Thus, under control-group stability, signal monotonicity implies
\begin{align*}
    \text{sign}(S_{i}^{T} - \phi(B_{i0}))\psi_{i}^{CT} = \text{sign}(S_{i}^{T} - \phi(B_{i0}))\text{sign}(\phi(B_{i1}^{T}(\cdot|S^{T}_{i})) - \phi(B_{i1}^{T}(\cdot|S^{\phi}_{i}))) \geq 0.
\end{align*}
In particular, with an appropriate TSLS specification, we can use the sign of the perception gap to correct for the contamination from $\psi_{i}^{CT}$. Thus, signal monotonicity and control-group stability correspond to ``weak'' IV monotonicity, which allows the direction of canonical IV monotonicity to vary across covariate values \citep{sloczynski2020should, blandhol2022tsls}.

The stability condition assumes that receiving information ``consistent'' with one's prior feature is equivalent to not receiving any information, in the sense that the posterior features are the same. For example, consider an agent that (i) has prior expectations of 3\% for inflation; and (ii) without additional information, maintains their prior beliefs. Stability assumes that this agent's posterior expectations remain at 3\% if given a signal that inflation will be 3\%.\footnote{Here it is fine if the agent's uncertainty decreases---we only need expectations to be stable.}

The stability condition restricts the signal to be directly comparable to the state of interest. For example, consider a design in which the elicited priors and posteriors are expectations $\mu$ over the extent of discrimination against Black individuals in the housing market, and the signal measures the extent of discrimination against Black individuals in the labor market. This treatment group plausibly satisfies signal monotonicity: It is reasonable that agents believe that higher levels discrimination against Black individuals in the labor market tend to imply higher levels of discrimination in the housing market. However, the value of $\mu(B_{i0})$ refers to a different state (labor market) than $S_{i}^{T}$ (housing market). In particular, being told that the level of discrimination in the housing market is $S_{i}^{\mu} := \mu(B_{i0})$ need not confirm one's priors over that discrimination.\footnote{An alternative design is to elicit prior expectations over the state corresponding to the signal, as is done in \cite{haaland2023racial}, who estimate an intent-to-treat effect. Our framework considers a common state space $\Omega$ for the prior and posterior for simplicity---to our knowledge, no existing paper that provides TSLS estimates uses the alternative design in \cite{haaland2023racial}.}

The stability condition limits the extent to which agents may update their beliefs beyond the signal provided, between the elicitation of priors and posteriors. For example, if agent $i$ searches for information with greater intensity when assigned to control than when assigned to treatment with a signal $s = S_{i}^{\phi}$ that ``confirms'' their priors, then $\phi(B_{i1}^{T}(\cdot|S^{\phi}_{i})) \neq \phi(B_{i1}^{C})$. However, in such cases it suffices to assume a weaker version of stability:
\begin{align*}
    |\phi(B_{i1}^{T}(\cdot|S^{\phi}_{i})) - \phi(B_{i1}^{C}(\cdot|S_{i}^{C}))| \leq |\phi(B_{i1}^{T}(\cdot|S_{i}^{T})) - \phi(B_{i1}^{T}(\cdot|S^{\phi}_{i}))|.
\end{align*}
This allows agents to acquire outside information \citep{cavallo2017inflation, armona2019home}, provided that the difference in $\phi$ from treatment with $s = S_{i}^{T}$ versus $s = S_{i}^{\phi}$ is large enough.

\subsection{Passive Control Weights}\label{main:sec:passive.parameters}
Proposition \ref{main:prop:passive.MLRP}, which follows from the analysis in Section \ref{main:sec:stability}, characterizes $w_{i}^{CT}$.

\begin{proposition}\label{main:prop:passive.MLRP}
Let Assumption \ref{main:ass:signal.monotonicity} be satisfied. If the passive control comparison is control-group stable and $S_{i}^{T} \neq \phi(B_{i0})$, then 
\begin{align*}
     w_{i}^{CT} &\propto |\phi(B_{i1}^{T}) - \phi(B_{i1}^{C})|, & I_{i}^{sign} &:= \text{sign}(S_{i}^{T} - \phi(B_{i0})), \\
     w_{i}^{CT} &\propto |\phi(B_{i1}^{T}) - \phi(B_{i1}^{C})| \times |S_{i}^{T} - \phi(B_{i0})|, & I_{i}^{gap} &:= S_{i}^{T} - \phi(B_{i0}), \\
     w_{i}^{CT} &\propto |\phi(B_{i1}^{T}) - \phi(B_{i1}^{C})| \times (\pi_{1}\psi_{i}^{CT} + \pi_{2}|S_{i}^{T} - \phi(B_{i0})|), & I_{i}^{1 + gap} &:= (1, S_{i}^{T} - \phi(B_{i0}))', \\
     w_{i}^{CT} &\propto |\phi(B_{i1}^{T}) - \phi(B_{i1}^{C})| \times (\pi_{1}\psi_{i}^{CT} + \pi_{2}\psi_{i}^{CT}\phi(B_{i0})), & I_{i}^{1 + prior} &:= (1, \phi(B_{i0}))'.
\end{align*} 
\end{proposition}
Thus, under signal monotonicity and control-group stability, only \textit{some} of the passive control specifications from the literature recover positive-weighted APEs. The weights from $I_{i}^{sign}$ are positive, and proportional to the absolute difference in $i$'s counterfactual posterior features across treatment and control.\footnote{If $S_{i}^{T} = \phi(B_{i0})$ for some $i$, then the weights for $I_{i}^{sign}$ are $w_{i}^{CT} \propto |\phi(B_{i1}^{T}) - \phi(B_{i1}^{C})|\1\{S_{i}^{T} \neq \phi(B_{i0})\}$.} Relative to $I_{i}^{sign}$, interaction $I_{i}^{gap}$ additionally up-weights agents with larger absolute perception gaps. Relative to $I_{i}^{sign}$ and $I_{i}^{gap}$, interactions $I_{i}^{1,gap}$ and $I_{i}^{1,prior}$ allow for negative weights. Intuitively, the signs of the scalars $I_{i}^{sign}$ and $I_{i}^{gap}$ correctly predict $\psi_{i}^{CT}$. On the other hand, the vectors $I_{i}^{1,gap}$ and $I_{i}^{1,prior}$ induce linear combinations $I_{i}'\pi$ that need not correctly predict $\psi_{i}^{CT}$.

\subsection{Neutrality}\label{main:sec:neutrality}
We now consider active control comparisons.

\begin{definition}[Neutrality]
An active control comparison is \textit{treatment-group neutral} if
\begin{align*}
    \phi(B_{i1}^{H}(\cdot|S^{L}_{i})) = \phi(B_{i1}^{L}(\cdot|S^{L}_{i})).
\end{align*}
\end{definition}
If an active control comparison is treatment-group neutral, then the provision of $s = S_{i}^{L}$ in both treatment groups leads agent $i$ to form the same posterior features. Notice
\begin{align*}
    \psi_{i}^{LH} = \text{sign}([\phi(B_{i1}^{H}(\cdot|S^{H}_{i})) - \phi(B_{i1}^{H}(\cdot|S^{L}_{i}))] + [\phi(B_{i1}^{H}(\cdot|S^{L}_{i})) - \phi(B_{i1}^{L}(\cdot|S^{L}_{i}))]).
\end{align*}
Thus, under treatment-group neutrality, signal monotonicity implies
\begin{align*}
    \text{sign}(S^{H}_{i} - S^{L}_{i})\psi_{i}^{LH} = \text{sign}(S^{H}_{i} - S^{L}_{i})\text{sign}(\phi(B_{i1}^{H}(\cdot|S^{H}_{i})) - \phi(B_{i1}^{H}(\cdot|S^{L}_{i}))) \geq 0.
\end{align*}
In particular, we can use $\text{sign}(S^{H}_{i} - S^{L}_{i})$ to correct for the contamination from $\psi_{i}^{LH}$. However, since $S_i^H > S_i^L$ for all $i$, then $\text{sign}(S^{H}_{i} - S^{L}_{i}) = 1$ so that existing active control specifications ``automatically'' make the correction. Therefore, signal monotonicity and treatment-group neutrality correspond to the canonical IV monotonicity condition \citep{imbens1994identification, angrist2000interpretation}.

Intuitively, the neutrality condition assumes that agents perceive the provision of information to be ``equally credible'' across treatment groups.\footnote{It suffices to assume that the belief updating rules $B_{i1}^{g}(\cdot|s)$ are the same across $g \in \{L, H\}$.} This equality may not hold when there are asymmetric effects of information provision across treatment groups. For example, agents may perceive information provision to be less credible in groups that provide information that more strongly challenges their prior beliefs \citep{gentzkow2006, haaland2023designing}. Such issues can be mitigated with appropriate experimental design, including framing the information as coming from similar sources and using the same presentation strategies across treatment groups. Furthermore, as with control-group stability, we can also entertain a weaker version of neutrality:
\begin{align*}
    |\phi(B_{i1}^{H}(\cdot|S^{L}_{i})) - \phi(B_{i1}^{L}(\cdot|S^{L}_{i}))| \leq |\phi(B_{i1}^{H}(\cdot|S^{H}_{i})) - \phi(B_{i1}^{H}(\cdot|S^{L}_{i}))|.
\end{align*}
This allows for some perceived asymmetry, provided that the difference in $\phi$ from treatment with $s = S_{i}^{H}$ versus $s = S_{i}^{L}$ is large enough.

\subsection{Active Control Weights}\label{main:sec:active.parameters}
Proposition \ref{main:prop:active.MLRP}, which follows from the analysis in Section \ref{main:sec:neutrality}, characterizes $w_{i}^{LH}$.

\begin{proposition}\label{main:prop:active.MLRP}
Let Assumption \ref{main:ass:signal.monotonicity} be satisfied. If the active control comparison is treatment-group neutral, then 
\begin{align*}
     w_{i}^{LH} \propto |\phi(B_{i1}^{H}) - \phi(B_{i1}^{L})|.
\end{align*}
\end{proposition}
In particular, under signal monotonicity and treatment-group neutrality, existing active control specifications recover positive-weighted APEs. The weights $w_{i}^{LH}$ are proportional to the absolute difference in agents' posterior features across treatment groups, which parallels $w_{i}^{CT}$ for the passive control specification with interaction $I_{i}^{sign}$. 

\section{Practical Recommendations}\label{main:sec:recs}

\subsection{Passive versus Active Control}
Given our framework, when deciding between active and passive control designs, researchers should consider how the associated comparisons $\phi(B_{i1}^{H}) - \phi(B_{i1}^{L})$ and $\phi(B_{i1}^{T}) - \phi(B_{i1}^{C})$ differ in (i) the causal parameters they recover; and (ii) the identification constraints they place on the design of information provision.

Within a given agent $i$, passive control aggregates the partial effects $\partial_{v}Y_{i}^{\phi}(v)$ over the range of values $v$ between $\phi(B_{i1}^C)$ and $\phi(B_{i1}^T)$, directly answering the question ``what is the effect of giving agent $i$ information, relative to not giving them information?'' In contrast, active control aggregates over values $v$ between $\phi(B_{i1}^L)$ and $\phi(B_{i1}^H)$, answering the question ``what is the effect of giving agent $i$ one piece of information relative to another?'' The latter counterfactual may have less policy relevance, as noted in \cite{haaland2023designing}.

Notably, even if $|\phi(B_{i1}^{T}) - \phi(B_{i1}^{C})| = |\phi(B_{i1}^{H}) - \phi(B_{i1}^{L})|$, we can still have $\beta^{CT} \neq \beta^{LH}$. For a simple illustration of this difference, let us assume that $S_i^T = S_i^H$, such that $B_{i1}^T = B_{i1}^C = B_{i1}^*$. If we trace out a partial effects curve $\partial_{v}Y_{i}^{\phi}(v)$ as in Figure \ref{fig:example}, then $\Bar{\beta}^{LH}_{i}$ is proportional to the green region and $\Bar{\beta}^{CT}_{i}$ is proportional to the orange region. Due to the shape of the partial effects curve for this agent, $\Bar{\beta}^{CT}_{i} > \Bar{\beta}^{LH}_{i}$. Thus, even if $|\phi(B_{i1}^{T}) - \phi(B_{i1}^{C})| = |\phi(B_{i1}^{H}) - \phi(B_{i1}^{L})|$ for all $i$, we can still have $\beta^{CT} \neq \beta^{LH}$.

\begin{figure}[h!]
    \centering
    \includegraphics[width=0.5\textwidth]{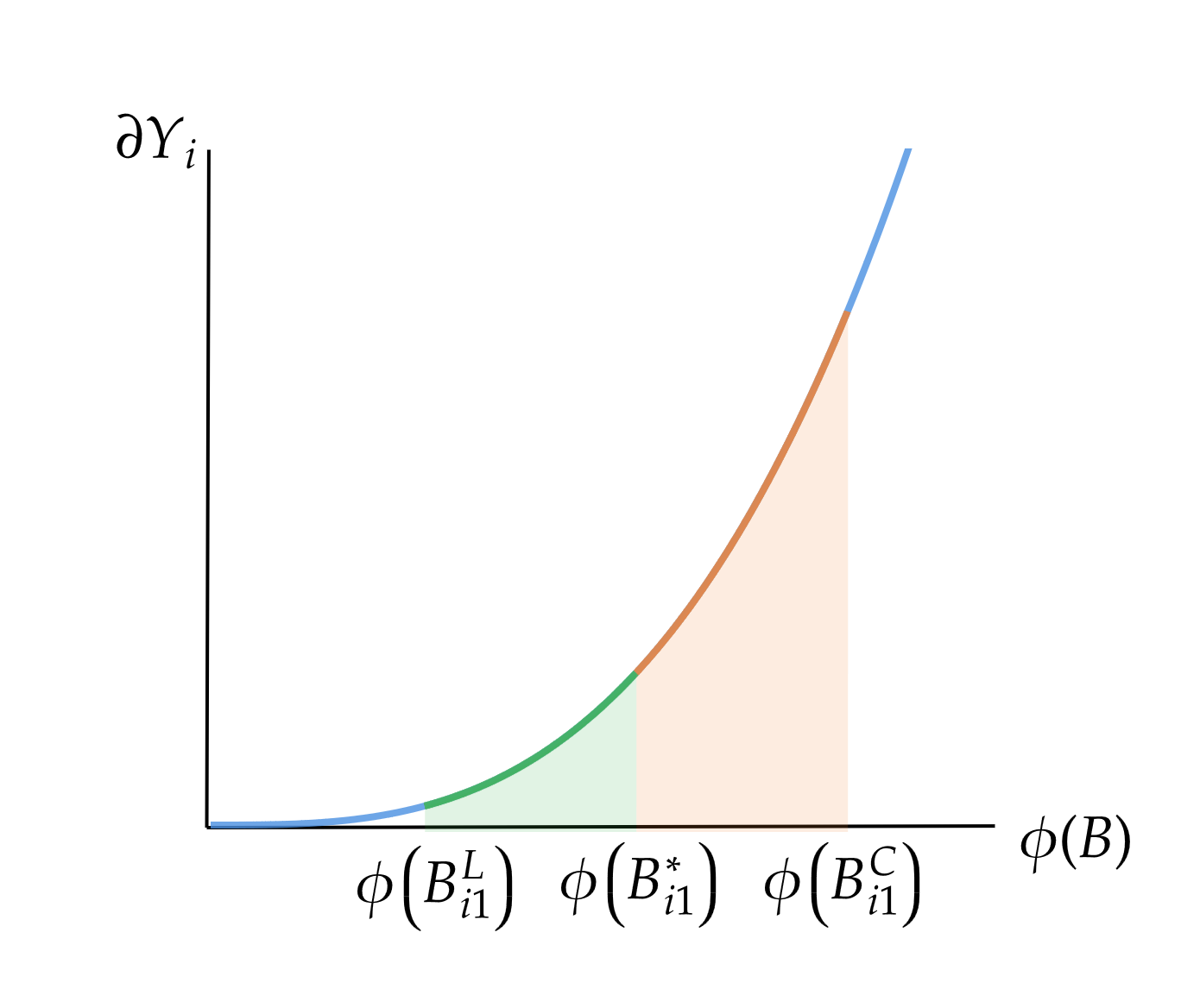}
    \caption{Comparing Active and Passive Control Counterfactuals}
    \label{fig:example}
    \floatfoot{\textit{Note:} The blue curve illustrates the partial effects curve $\partial_{v}Y_{i}^{\phi}(v)$ for $i$. We let $B_{i1}^* = B_{i1}^T = B_{i1}^H$. Therefore, the green region is proportional to $\Bar{\beta}^{LH}_{i}$---the within-agent APE associated with the active control comparison for agent $i$. The orange region is proportional to $\Bar{\beta}^{CT}_{i}$---the within-agent APE associated with the passive control comparison for agent $i$.} 
\end{figure}

Passive and active control also differ in how they aggregate partial effects across agents. The latter weights agents according to $|\phi(B_{i1}^H) - \phi(B_{i1}^L)|$, while the former with interaction $I_i^{sign}$ weights agents according to $|\phi(B_{i1}^T) - \phi(B_{i1}^C)|$.\footnote{We focus on $I_i^{sign}$ because it induces weights $w_{i}^{CT}$ that are most analogous to $w_{i}^{LH}$.} Because agents assigned to control are not given information, $|\phi(B_{i1}^T) - \phi(B_{i1}^C)|$ is more sensitive to agents' priors. When $|\phi(B_{i1}^T) - \phi(B_{i1}^C)|$ is more heterogeneous than $|\phi(B_{i1}^H) - \phi(B_{i1}^L)|$, passive control places less equal weights across agents; this is more likely in settings with highly  heterogeneous priors---see Appendix \ref{main:sec:gaussian.model}.  Therefore, although active control generates a potentially less policy-relevant counterfactual, it may place more equal weights across agents, which may be desirable in some settings \citep{balla2023identifying}. This analysis informs discussions from \cite{roth2020expectations}
and \cite{haaland2023designing}, which posit that active control identifies causal effects for a ``broader population'' of agents. This claim is not exactly true; rather, it is the relative weight on agents with differing priors that active control may better equalize.

Finally, active control requires treatment-group neutrality, whereas passive control requires control-group stability. Under neutrality, $S_i^L$ and $S_i^H$ must be perceived to be equally credible; this imposes design constraints on how distinct the two signals can be (since more extreme signals may lead to loss of credibility), and what sources may be used to generate the signals. Under stability, $S_i^T$ must convey similar numerical content to the prior feature (i.e., same units for the signals as the prior features), and agents should not differentially search for additional information between prior and posterior elicitation, dependent on group assignment. 

\subsection{Choosing a Passive Control Specification}
Among the four passive control specifications, interactions $I_i^{sign}$ and $I_i^{gap}$ generate positive weights, whereas $I_i^{1,gap}$ and $I_i^{1,prior}$ do not. Therefore, the latter two specifications should not be used: Agents' partial effects may enter \textit{negatively} into the weighted average, attenuating or reversing the sign of the estimated APE. 

Between $I_i^{sign}$ and $I_i^{gap}$, the former directly maps to canonical formulations of the local average treatment effect for continuous endogenous variables \citep{angrist2000interpretation}, which weights by the size of agents' ``first-stages.'' In this setting, these agent first-stages are given by $|\phi(B_{i1}^{T}) - \phi(B_{i1}^{C})|$. However, if researchers are especially interested in agents with large perception gaps, then $I_i^{gap}$ may generate a more useful APE. Moreover, $I_i^{sign}$ and  $I_i^{gap}$ may provide different levels of estimation precision. 

\section{Applications}\label{main:sec:applications}
We now compare the different passive control interactions using data from \cite{kumar2023effect} and \cite{jager2022worker}. We show that the chosen specification affects the magnitude and significance of estimates, and provide evidence that these effects are driven by contaminated weights in the former setting, and by up-weighting of agents with larger perception gaps in the latter setting. 

\subsection{Kumar et al. (2023)}
\cite{kumar2023effect} study how firms' expectations and uncertainty of future GDP growth affect their economic decisions. We compare two arms in their experiment: the control arm, which receives no information, and their first treatment arm, which tells firms that the average prediction of GDP growth among a panel of professional forecasters is 4\%. \cite{kumar2023effect} elicit agents' expectations and uncertainty---we focus on expectations.\footnote{As noted in Section \ref{main:sec:setup}, econometric issues arise when there is more than one endogenous variable; hence, we avoid including both expectations and uncertainty.} The outcome variable $Y_{i}$ is defined as the difference between realized and planned changes in business choices, such as prices and employment. 

Figure \ref{fig:kumar-coeff} plots TSLS coefficients for various outcomes: price, employment, and advertising budget. Other outcomes are given in Appendix \ref{main:sec:appendix.applications.1}. Since all firms in the treatment group receive the same professional forecast, interactions $I_i^{1,gap}$ and $I_i^{1,prior}$ provide the same estimates.

\begin{figure}[h!]
    \centering
    \includegraphics[width = \textwidth]{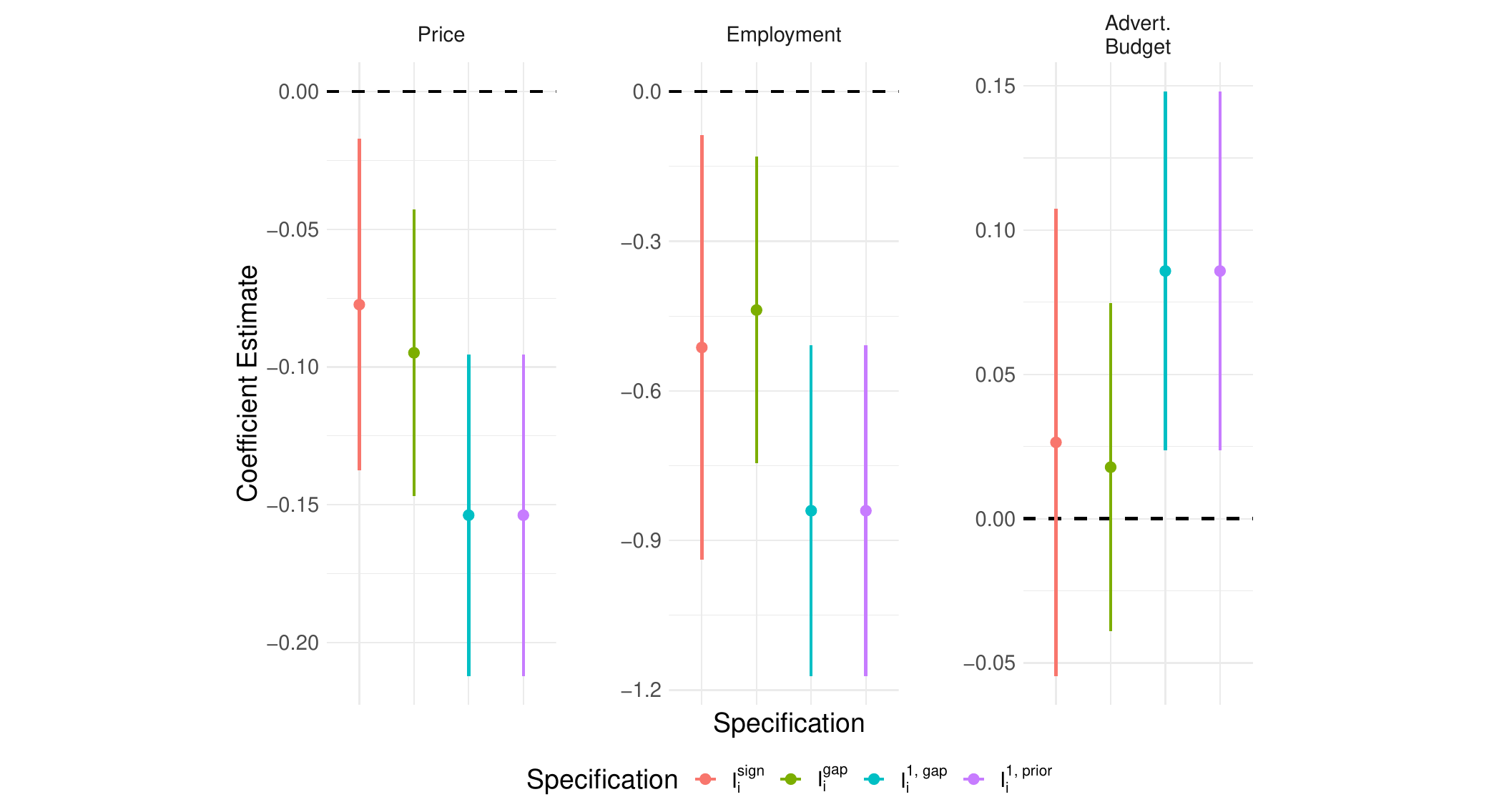}
    \caption{Coefficient Estimates---Data from \cite{kumar2023effect}}
    \label{fig:kumar-coeff}
    \floatfoot{\textit{Note:} The points are point estimates and the bars are 95\% confidence intervals. The outcomes $Y_{i}$ are defined as the difference between the planned change and the actualized change, i.e., the planned change in price versus the actualized changed in price. Price is the price of the firm's main product, employment is total employment at the firm, and advert. budget is the advertising budget of the firm. The signal $S_{i}^{T}$ is the average prediction of 4\% GDP growth among professional forecasters, the prior feature $\mu(B_{i0})$ is pre-treatment GDP growth expectations, and the posterior feature $\mu(B_{i1})$ is post-treatment GDP growth expectations. $I_i^{sign}$ regresses the outcome on the sign of the perception gap and posterior GDP growth expectations, instrumenting the posterior by the treatment indicator times the sign of the perception gap $S_i^T - \mu(B_{i0})$. $I_i^{gap}$ regresses the outcome on the perception gap $S_i^T - \mu(B_{i0})$ and the posterior, instrumenting the posterior by the treatment indicator times the perception gap. $I_i^{1, gap}$ regresses the outcome on the perception gap and the posterior, instrumenting the posterior by the treatment indicator and the treatment indicator times the perception gap. $I_i^{1, prior}$ regresses the outcome on the prior and the posterior, instrumenting the posterior by the treatment indicator and the treatment indicator times the prior. In all specifications, the coefficient of interest is the coefficient on the posterior expectation.}
\end{figure}

For price, the magnitude of the coefficient for interaction $I_i^{sign}$ is halved and the magnitude of the coefficient for interaction $I_i^{gap}$ is reduced by one-third, relative to the coefficients for interactions $I_i^{1,gap}$ and $I_i^{1,prior}$. A t-test of these coefficients finds that $I_i^{sign}$ is marginally significantly different (p-value = 0.078) from $I_i^{1,gap}$ and $I_i^{1,prior}$. Similarly, for employment, the magnitude of the coefficients for $I_i^{sign}$ and $I_i^{gap}$ are 60\% and 50\% of the other coefficients; $I_i^{gap}$ is marginally significantly different (p-value = 0.081) from $I_i^{1,gap}$ and $I_i^{1,prior}$. Finally, for advertising budget, the magnitudes of the coefficients for $I_i^{sign}$ and $I_i^{gap}$ are a third of the coefficients for interactions $I_i^{1,gap}$ and $I_i^{1,prior}$. In fact, under $I_i^{sign}$ and $I_i^{gap}$, the 95\% confidence interval on the advertising budget estimate include zero effects.

What explains the above differences? Recall that interactions $I_i^{sign}$ and $I_i^{gap}$, which estimate coefficients of similar magnitude, both generate positive weights. In contrast, interactions $I_i^{1,gap}$ and  $I_i^{1,prior}$, which estimate coefficients that are twice as large, both allow for negative weights. Although we cannot formally reject the non-existence of negative weights for interactions $I_i^{1,gap}$ and  $I_i^{1,prior}$, the above grouping pattern suggests that negative weights may be driving larger TSLS estimates for $I_i^{1,gap}$ and $I_i^{1,prior}$. 

\subsection{Jäger et al. (2024)}
\cite{jager2022worker} study how information about workers' outside options---that is, workers' wages if they left their current job to find a new one---affects their labor market decisions. First, \cite{jager2022worker} elicit workers' prior expectations of their outside options. Next, they split their sample into a control group and a single treatment group. The control arm receives no information; those in the treatment group are told the mean wage of workers similar to themselves (based on gender, age, occupation, labor market region, and education level). Then, \cite{jager2022worker} elicit workers' posterior expectations of their outside options.

Figure \ref{fig:jager-coeff} displays estimated coefficients from all four interactions using intended negotiation probability as the outcome---estimates for other outcomes are in Appendix \ref{main:sec:appendix.applications.1}. The interaction $I_i^{sign}$ coefficient is 30-40\% larger (though not statistically significantly so) than the coefficients from interactions $I_i^{gap}$, $I_i^{1,gap}$, and $I_i^{1,prior}$, each of which are similar to each other. The fact that interaction $I_i^{gap}$, which ensures positive weights, is similar to interactions $I_i^{1,gap}$ and $I_i^{1,prior}$ suggests that the potential negative weighting in the latter is not consequential in this setting. Instead, the smaller coefficients for interactions $I_i^{gap}$, $I_i^{1,gap}$, and $I_i^{1,prior}$ may be due to the up-weighting of agents with larger perception gaps. We provide further evidence that up-weighting may be driving these results in Appendix \ref{main:sec:appendix.applications.2}. 

\begin{figure}[h!]
    \centering
    \includegraphics[width = 0.7\textwidth]{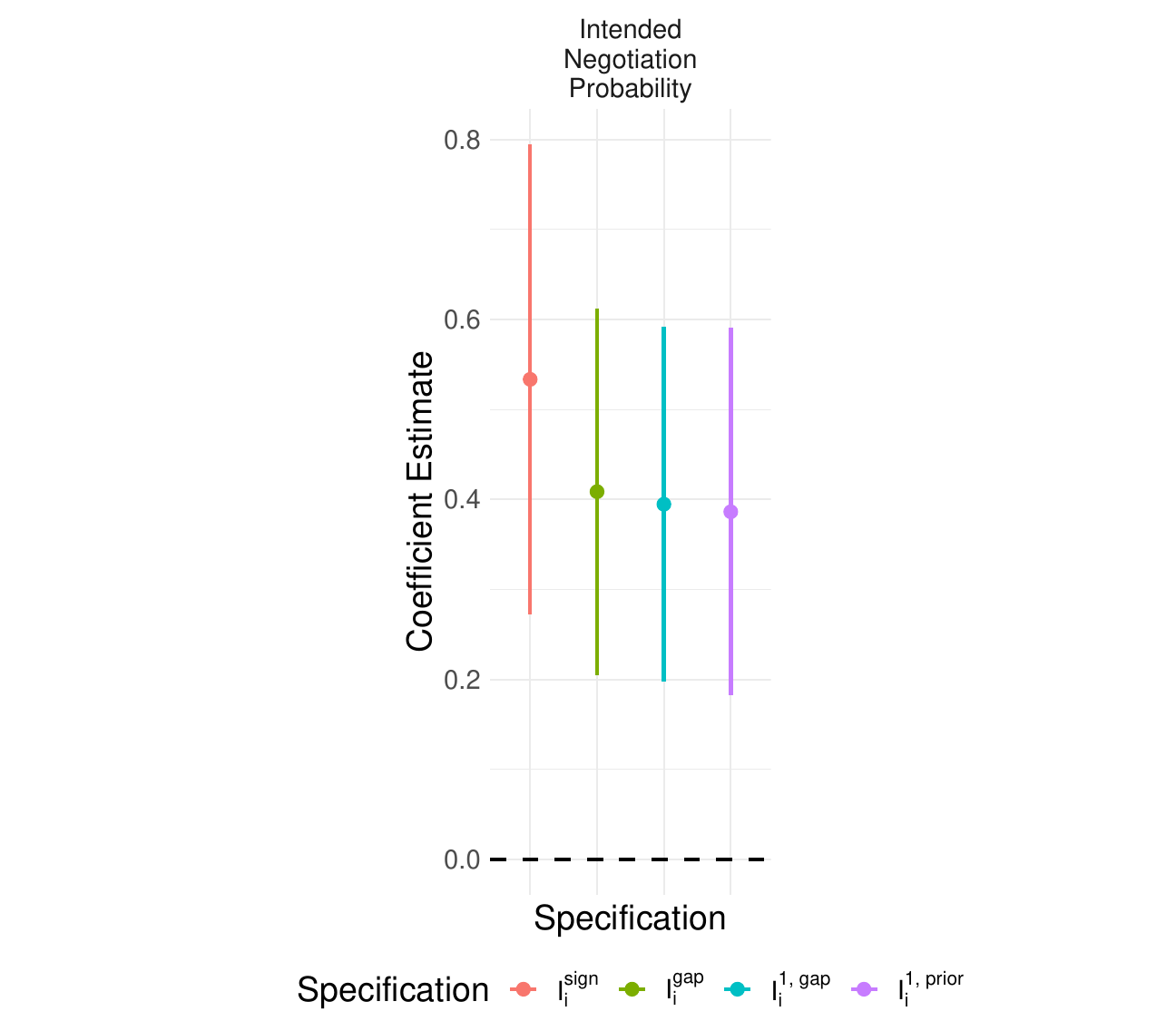}
    \caption{Coefficient Estimates---Data from \cite{jager2022worker}}
    \label{fig:jager-coeff}
    \floatfoot{\textit{Note:} The points are point estimates and the bars are 95\% confidence intervals. The outcome is the intended probability of negotiating for a raise. The signal $S_{i}^{T}$ is the mean wage of workers similar to worker $i$, and the prior $\mu(B_{i0})$ and posterior $\mu(B_{i1})$ are worker $i$'s expectations of how much their wages would change, as a percentage of their current wage, pre- and post- treatment, respectively. $I_i^{sign}$ regresses the outcome on the sign of the perception gap and the posterior, instrumenting the posterior by the treatment indicator times the sign of the perception gap. $I_i^{gap}$ regresses the outcome on the perception gap and the posterior, instrumenting the posterior by the treatment indicator times the perception gap. $I_i^{1, gap}$ regresses the outcome on the perception gap and the posterior, instrumenting the posterior by the treatment indicator and the treatment indicator times the perception gap. Following \cite{jager2022worker}, we normalize agents' perception gaps $S_{i}^{T} - \mu(B_{i0})$ to be a percentage of $S_{i}^{T}$ for $I_i^{gap}$ and $I_i^{1, gap}$. $I_i^{1, prior}$ regresses the outcome on the prior and the posterior, where the posterior is instrumented by the treatment indicator, the treatment indicator times the prior, and the treatment indicator times the signal. We estimate the version of interaction $I^{1,prior}_{i}$ that includes both the signal and the prior because signals are personalized to each agent---see Interaction \ref{main:int:kumar} for a discussion. In all specifications, the coefficient of interest is the coefficient on the posterior expectation. 
} 
\end{figure}

One explanation for why up-weighting may have attenuated the estimates is endogenous information acquisition, such as in \cite{balla2023identifying}. Agents who are more responsive to information are more likely to seek it out, and therefore have more accurate priors. At the same time, agents that are very accurate will not respond much to the information treatment; their actions are already near-optimal. Therefore, the agents that are most responsive are those that have medium-sized perception gaps, and the agents that are least responsive are those that have large or no perception gap. Interactions $I_i^{gap}$, $I_i^{1,gap}$, and $I_i^{1,prior}$ will up-weight agents with large perception gaps---and thus those who are least responsive---attenuating treatment effects towards zero.

\section{Conclusion}
This paper examines identification and estimation in information provision experiments. We developed a causal framework for these experiments, and provided conditions for recovering interpretable APEs in both passive and active control designs. Using this framework, we evaluated TSLS specifications from the literature. We found that the standard active control specification identifies an APE with non-negative weights. However, two of the four common passive control specifications allow for negative weights, which we advise against.

\newpage
\bibliography{lit.bib}
\newpage

\section*{Appendix}
\renewcommand{\thesubsection}{\Alph{subsection}}
\renewcommand{\thelemma}{\Alph{subsection}\arabic{lemma}}

\subsection{Posterior Formation}\label{main:sec:posterior.formation}

Assume that $\Omega \subseteq \R$. In settings where the state space is multi-dimensional, we can apply the arguments below to the relevant marginal distribution of $B$ instead. Let $\nu \in \Delta(\R)$ be a sigma-finite measure.

\subsubsection{Perceptions}
We place structure on how agents perceive the information provision. Formally, agents in treatment group $g$ perceive the signals as realizations from signal distributions $\{Q_{i1}^{g}(\cdot|\omega): \omega \in \Omega\} \subseteq \Delta(\R)$ that are informative for the underlying state $\omega \in \Omega$. These distributions have densities $q_{i1}^{g}(\cdot|\omega) := dQ_{i1}^{g}(\cdot|\omega)/d\nu$ with respect to $\nu$.
\begin{definition}[MLR]\label{main:def:MLRP}
A family of perceived signal distributions $\{Q_{i1}^{g}(\cdot|\omega): \omega \in \Omega\} \subseteq \Delta(\R)$ satisfies the \textit{monotone likelihood ratio} (MLR) property in $s \in \R$ if $\omega' > \omega$ and $s' > s$ implies
\begin{align*}
    \frac{q_{i1}^{g}(s'|\omega')}{q_{i1}^{g}(s'|\omega)} \geq \frac{q_{i1}^{g}(s|\omega')}{q_{i1}^{g}(s|\omega)}.
\end{align*}
\end{definition}
This property (c.f. \citet[Definition 8.3.16]{casella2021statistical}) is satisfied for a number of signal distribution families, including the set of Gaussian beliefs considered in the literature---see Section \ref{main:sec:gaussian.model}. In fact, many exponential families satisfy Definition \ref{main:def:MLRP} in their respective sufficient statistics \citep[Section 4.5]{bernardo2009bayesian}. Intuitively, agents assigned to treatment group $g$ assume that the experiment tends to provide larger signals $s$ under larger realizations of the state $\omega \in \Omega$. 

\subsubsection{Belief Updating}
To link the above MLR property to signal monotonicity, we must consider the belief updating rules $s \mapsto B_{i1}^{g}(\cdot|s)$ in greater detail. We suppose that beliefs $B \in \B$ have densities $b := dB/d\nu$ with respect to $\nu$. In particular, the priors $B_{i0}$ have densities $b_{i0} := dB_{i0}/d\nu$. We first consider Bayesian rules.

\begin{definition}[Bayesian]
$s \mapsto B_{i1}^{g}(\cdot|s)$ is \textit{Bayesian} if
\begin{align*}
    \frac{dB_{i1}^{g}(\omega|s)}{d\nu} = \frac{q_{i1}^{g}(s|\omega)b_{i0}(\omega)}{\displaystyle \int q_{i1}^{g}(s|\omega)b_{i0}(\omega) d\nu(\omega)}.
\end{align*}
\end{definition}
These Bayesian rules are prevalent in models from the literature \citep{armantier2016price, cavallo2017inflation, armona2019home, cullen2022much, fuster2022expectations, balla2022determinants}. However, the MLR framework can accommodate systematic deviations from Bayesian updating. 

\begin{definition}[Anchored]
$s \mapsto B_{i1}^{g}(\cdot|s)$ is \textit{anchored} if there exists $\tau_{i}^{g} \in [0,1]$ and $\Bar{B}_{i}^{g} \in \B$ such that
\begin{align*}
    \frac{dB_{i1}^{g}(\omega|s)}{d\nu} = \tau_{i}^{g}\frac{d\Bar{B}_{i}^{g}(\omega)}{d\nu} + (1 - \tau_{i}^{g})\frac{q_{i1}^{g}(s|\omega)b_{i0}(\omega)}{\displaystyle \int q_{i1}^{g}(s|\omega)b_{i0}(\omega) d\nu(\omega)}.
\end{align*}
\end{definition}
Anchored rules allows agents' updating rules to be distorted away from the Bayesian baseline, and towards some anchor belief $\Bar{B}_{i}^{g} \in \B$. The anchored rules accommodate numerous behavioral biases from the behavioral economics literature \citep[Section 2.3]{gabaix2019behavioral}, including the anchoring-and-adjustment heuristics discussed in \cite{tversky1974judgment}. For instance, choosing the anchors to be the prior beliefs $\Bar{B}_{i}^{g} = B_{i0}$ corresponds to \textit{conservatism}, wherein agents insufficiently update their beliefs relative to the Bayesian baseline \citep{edwards1968conservatism}.

The anchored updating rules are ``affine'' distortions of Bayesian updating \citep{de2022non}. However, signal monotonicity also accomodates ``nonlinear'' distortions. The following class of nonlinear distortions is due to \cite{grether1980bayes}. 

\begin{definition}[Grether]
$s \mapsto B_{i1}^{g}(\cdot|s)$ is \textit{Grether} if there exist $\chi_{i0}^{g}, \chi_{i1}^{g} > 0$ such that
\begin{align*}
    \frac{dB_{i1}^{g}(\omega|s)}{d\nu} = \frac{q_{i1}^{g}(s|\omega)^{\chi_{i1}^{g}}b_{i0}(\omega)^{\chi_{i0}^{g}}}{\displaystyle \int q_{i1}^{g}(s|\omega)^{\chi_{i1}^{g}}b_{i0}(\omega)^{\chi_{i0}^{g}} d\nu(\omega)}.
\end{align*}
\end{definition}

The inference and base-rate parameters $\chi_{i0}^{g}, \chi_{i1}^{g} > 0$ allow for flexible deviations from Bayesian updating---the discussion below closely follows \citet[page 103]{benjamin2019errors}. On the one hand, $\chi_{i1}^{g} < 1$ ($\chi_{i1}^{g}>1$) reflects under-inference (over-inference) in the sense that agents in treatment group $g$ update as if the signals $s$ are less (more) information for the states $\omega \in \Omega$ than in the Bayesian baseline of $\chi_{i1}^{g} = 1$. On the other hand, $\chi_{i0}^{g} < 1$ ($\chi_{i0}^{g} > 1$) reflects base-rate neglect (over-use) in the sense that agents in treatment group $g$ update as if their priors $B_{i0}$ are less (more) information for the states $\omega \in \Omega$ than in the Bayesian baseline of $\chi_{i0}^{g} = 1$. 

\begin{remark}
Allowing for non-Bayesian updating has practical relevance, since the information provision literature often acknowledges the potential for behavioral biases in belief formation. For example, there is concern that agents may numerically anchor their posterior expectations at the values of quantitative signals \citep{cavallo2017inflation}. This behavior corresponds to anchored updating with $\Bar{B}_{i}^{g}$ as the distribution that places unit mass on the provided signal. There is also concern that information provision induces emotional responses \citep{haaland2023designing}. If such responses are suggestive of motivated reasoning, then we can microfound the anchoring parameters $\tau_{i}^{g} \in [0,1]$ with a model of utility-maximization \citep[Example 2]{de2022non}; similar arguments apply for the inference and base-rate parameters $\chi_{i0}^{g}, \chi_{i1}^{g}$ in the Grether case.
\end{remark}

\subsubsection{Signal Monotonicity via MLR}
The following proposition shows that signal monotonicity is satisfied in the MLR framework.
\begin{proposition}\label{main:prop:MLRP}
For each treatment group $g$, suppose that $\{Q_{i1}^{g}(\cdot|\omega): \omega \in \Omega\} \subseteq \Delta(\R)$ satisfies the MLR property in $s \in \R$ and that $s \mapsto B_{i1}^{g}(\cdot|s)$ is either Bayesian, anchored, or Grether. Then Assumption \ref{main:ass:signal.monotonicity} is satisfied for any $B \mapsto \phi(B) := \int \varphi(\omega) dB(\omega)$ such that $\omega \mapsto \varphi(\omega)$ is increasing.
\end{proposition}

Proposition \ref{main:prop:MLRP} shows that signal monotonicity is satisfied under general conditions on agents' posterior formation. We allow for rich heterogeneity across agents: some can have Bayesian updating rules, while others have anchored or Grether rules with agent-specific distortion parameters. We also do not restrict agents' prior beliefs, opting instead to impose the structure of the MLR property on the response of beliefs to information provided.

\begin{remark}
Proposition \ref{main:prop:MLRP} covers expectations and second moments, which correspond to $\varphi(\omega) = \omega$ and $\varphi(\omega) = \omega^{2}$, respectively. However, it does not ensure signal monotonicity for the variance, which is the difference of second moments and the squared-expectations. That said, it is important to note that the MLR framework only gives a set of \textit{sufficient} conditions for Assumption \ref{main:ass:signal.monotonicity}. In particular, signal monotonicity can still hold for the variance, as discussed in Section \ref{main:sec:signal.monotonicity}.
\end{remark}

\subsubsection{Gaussian Model}\label{main:sec:gaussian.model}
Our framework nests models of Bayesian learning with Gaussian beliefs from the literature. Relative to this model, we allow for richer patterns of heterogeneity in agents' prior beliefs, perceptions of information provision, and belief updating rules.

In the Gaussian model, agents are Bayesian with Gaussian prior beliefs. Agent $i$ assumes that, if the true state is $\omega \in \Omega$, then $S_{i}^{g}$ for each treatment group $g$ is a realization from a Gaussian distribution with mean $\omega$ and variance $\varsigma_{i}^{2}$. Proposition \ref{main:prop:MLRP} applies---in this case, $B_{i1}^{g}(\cdot|s) = B_{i1}(\cdot|s) $ is Gaussian for each $s \in \R$, with mean and variance given by
\begin{align}\label{main:eq:linear.gaussian}
\begin{split}
    \mu(B_{i}(\cdot|s)) &= r_{i}s + (1 - r_{i})\mu(B_{i0}), \\
    \sigma^{2}(B_{i}(\cdot|s)) &= (1 - r_{i})\sigma^{2}(B_{i0}),
\end{split}
&  r_{i} &= \frac{\sigma^{2}(B_{i0})}{\sigma^{2}(B_{i0}) + \varsigma_{i}^{2}}.
\end{align}
In particular, the posterior expectation $\mu(B_{i}^{g})$ is a weighted average of the treatment signal $S_{i}^{g} \in \R$ and the prior mean $\mu(B_{i0})$, with weights $r_{i} \in (0,1)$ and $1 - r_{i}$ that are proportional to agents' certainty $1/\sigma^{2}(B_{i0})$ in their prior beliefs, and the perceived precision $1/\varsigma_{i}^{2}$ of the signal distribution---for a derivation, see \citet[Section 5]{hoff2009first}. In the information provision literature, $r_{i}$ is known as the \textit{learning rate} for agent $i$. These learning rates are higher the more the signals are perceived as informative for the underlying state. This model automatically imposes treatment-group neutrality. If we also assume $\mu(B_{i1}^{C}) = \mu(B_{i0})$ to ensure control-group stability, then we obtain
\begin{align*}
    |\mu(B_{i1}^{T}) - \mu(B_{i1}^{C})| &=  r_{i}|S_{i}^{T} - \mu(B_{i0})| \\
    |\mu(B_{i1}^{H}) - \mu(B_{i1}^{L})| &=  r_{i}|S_{i}^{H} - S_{i}^{L}|.
\end{align*}
Notice that the both of these quantities depend on priors through $r_{i}$, but the passive control comparison further depends on priors via the perception gap.

\subsection{Alternative Formulations}\label{main:sec:alt.formulations}
In the main text, we formulated the causal parameters of interest in terms of partial effects: $\partial_{v}Y_{i}^{\phi}(v)$. In this section, we consider alternative formulations for behavioral elasticities and discrete actions. To see how we can proceed, recall that under Assumption \ref{main:ass:IV}, together with rank and relevance conditions, Lemma \ref{main:lemma:fully.general} implies that the considered TSLS specifications recover estimands of the form
\begin{align*}
    \gamma = \frac{\E[I_{i}'\pi(Y_{i}(B_{i1}^{\tilde{g}})- Y_{i}(B_{i1}^{g}))]}{\E[I_{i}'\pi(\phi(B_{i1}^{\tilde{g}})- \phi(B_{i1}^{g}))]}.
\end{align*}
Assumption \ref{main:ass:parametric} allowed us to appeal to Lemma \ref{main:lemma:FTC} à la
\begin{align*}
    Y_{i}(B_{i1}^{\tilde{g}})- Y_{i}(B_{i1}^{g}) &= Y_{i}^{\phi}(\phi(B_{i1}^{\tilde{g}}))- Y_{i}^{\phi}(\phi(B_{i1}^{g}))\\
    &= \psi_{i}^{g\Tilde{g}}|\phi(B_{i1}^{\tilde{g}}) - \phi(B_{i1}^{g})|\int \lambda_{i}^{g\Tilde{g}}(v)\partial_{v}Y_{i}^{\phi}(v)dv,
\end{align*}
where $v \mapsto \lambda_{i}^{g\Tilde{g}}(v)$ is the density function for the uniform distribution over the feature values between $\phi(B_{i1}^{g})$ and $\phi(B_{i1}^{\tilde{g}})$. We can modify this argument to handle behavioral elasticities and discrete actions. In what follows, let $\V_{i}^{g\Tilde{g}}$ denote the set of feature values between $\phi(B_{i1}^{g})$ and $\phi(B_{i1}^{\tilde{g}})$.

\subsubsection{Behavioral Elasticities}\label{main:sec:elasticities}
If $Y_{i}(B) \neq 0$ and $\phi(B) \neq 0$, then we can formulate the causal parameters of interest in terms of partial elasticities---these \textit{behavioral elasticities} generate unit-free measures, which facilitates comparisons across applications \citep[Section 8]{haaland2023designing}. We therefore consider TSLS specifications that replace $Y_{i}$ and $\phi(B_{i1})$ with $\log(Y_{i}^{n})$ and $\log(\phi(B_{i1})^{n})$, for some fixed $n \in \mathbb{N}$ such that $Y_{i}(B)^{n} > 0$ and $\phi(B)^{n} > 0$. Under Assumption \ref{main:ass:parametric}, Lemma \ref{main:lemma:FTC} implies
\begin{align*}
    \log(Y_{i}(B_{i1}^{\tilde{g}})^{n}) - \log(Y_{i}(B_{i1}^{g})^{n}) &= \log(Y_{i}^{\phi}(\phi(B_{i1}^{\tilde{g}}))^{n}) - \log(Y_{i}^{\phi}(\phi(B_{i1}^{g}))^{n})\\
    &= n\psi_{i}^{g\Tilde{g}}\int v\frac{\partial_{v}Y_{i}^{\phi}(v)}{Y_{i}^{\phi}(v)} v^{-1}\1\{v \in \V_{i}^{g\Tilde{g}}\}dv,
\end{align*}
and
\begin{align*}
    \log(\phi(B_{i1}^{\tilde{g}})^{n}) - \log(\phi(B_{i1}^{g})^{n}) &= n\psi_{i}^{g\Tilde{g}} \int v^{-1}\1\{v \in \V_{i}^{g\Tilde{g}}\} dv.
\end{align*}
Therefore, we obtain
\begin{align*}
    \gamma = \E\left[\displaystyle \int \Lambda_{i}^{g\tilde{g}}(v)\, v\frac{\partial_{v}Y_{i}^{\phi}(v)}{Y_{i}^{\phi}(v)} dv\right], \quad \Lambda_{i}^{g\tilde{g}}(v) &= \frac{\psi_{i}^{g\tilde{g}}I_{i}'\pi \, v^{-1}\1\{v \in \V_{i}^{g\tilde{g}}\}}{\E\left[\displaystyle \int \psi_{i}^{g\tilde{g}}I_{i}'\pi \, v^{-1}\1\{v \in \V_{i}^{g\tilde{g}}\} dv \right]}.
\end{align*}

The above is a weighted average of partial elasticities $v\partial_{v}Y_{i}^{\phi}(v)/Y_{i}^{\phi}(v)$ with weights $\Lambda_{i}^{g\tilde{g}}(v)$ that integrate to one across $v$ and $i$. Using the log-version of Specification \ref{main:spec:active} corresponds to $I_{i} = 1$, and so given treatment-group neutrality, we obtain $\psi_{i}^{LH} = 1$ such that
\begin{align*}
    \Lambda_{i}^{LH}(v) = \frac{v^{-1}\1\{v \in \V_{i}^{LH}\}}{\E\left[\displaystyle \int v^{-1}\1\{v \in \V_{i}^{LH}\} dv \right]} \geq 0.
\end{align*}
In similar fashion, we can consider the log-version of Specification \ref{main:spec:passive} with Interaction \ref{main:int:IPIV}. Given control-group stability and $S_{i}^{T} \neq \phi(B_{i0})$, we obtain $\psi_{i}^{CT}I_{i}^{sign}\1\{v \in \V_{i}^{CT}\} = \1\{v \in \V_{i}^{CT}\}$ such that
\begin{align*}
    \Lambda_{i}^{CT}(v) = \frac{v^{-1}\1\{v \in \V_{i}^{CT}\}}{\E\left[\displaystyle \int v^{-1}\1\{v \in \V_{i}^{CT}\} dv \right]} \geq 0.
\end{align*}

\subsubsection{Discrete Actions}\label{main:sec:discrete.actions}
To consider the case of discrete actions, we modify Assumption \ref{main:ass:parametric} as follows. We assume that $\{\phi(B): B \in \B\}$ is a convex set, and that there exists a continuously differentiable $v \mapsto \Bar{Y}_{i}(v)$ defined over it such that $B \mapsto \E[Y_{i}(B)|R_{i}] = \Bar{Y}_{i}^{\phi}(\phi(B))$, where $R_{i}$ is some random vector. If $I_{i}$ is a function or subset of $R_{i}$, then
\begin{align*}
    \E[I_{i}'\pi(Y_{i}(B_{i1}^{\tilde{g}})- Y_{i}(B_{i1}^{g}))] &= \E[I_{i}'\pi [\Bar{Y}_{i}^{\phi}(\phi(B_{i1}^{\tilde{g}})) - \Bar{Y}_{i}^{\phi}(\phi(B_{i1}^{g}))]] \\
    &= \E\left[\psi_{i}^{g\Tilde{g}}I_{i}'\pi|\phi(B_{i1}^{\tilde{g}}) - \phi(B_{i1}^{g})|\int \lambda_{i}^{g\Tilde{g}}(v)\partial_{v}\Bar{Y}_{i}^{\phi}(v)dv\right].
\end{align*}
Therefore, we arrive at similar expressions for $\gamma$ as in Specifications \ref{main:spec:passive} and \ref{main:spec:active}, except now the partial effects of features on actions $\partial_{v}Y_{i}^{\phi}(v)$ are replaced by the partial effects of features on average actions $\Bar{Y}_{i}^{\phi}(v)$ conditional on $R_{i}$. The latter could entirely consist of the observed ``pre-determined'' variables in the experiment (i.e., prior features, signals, and agent characteristics). However, we also allow for latent unobserved components. The following example should fix ideas.

Let $\Omega = \R$ and consider binary actions $B \mapsto Y_{i}(B) \in \{0,1\}$ that satisfy
\begin{align*}
    Y_{i}(B) = \argmax_{y \in \{0,1\}} \int y \omega + (1 - y) \xi_{i} \, dB(\omega), \quad B \in \B.
\end{align*}
In particular, agents maximize their subjective expected utility---we can think of $\xi_{i}$ as an outside option, and $\omega$ as the benefit to choosing $y=1$. The above leads to
\begin{align*}
    Y_{i}(B) = \1\{\mu(B) \geq \xi_{i}\}.
\end{align*}
If $\xi_{i}$ has a continuous distribution conditional on $R_{i}$, then we can use $\Bar{Y}_{i}^{\phi}(v) = \P(\xi_{i} \leq v|R_{i})$.

\renewcommand{\theassumption}{\Alph{subsection}\arabic{assumption}}
\setcounter{assumption}{0}

\subsection{General Identification}\label{main:sec:general.identification}
There are groups $g \in \G$, a set of elicited features $\Phi$, and signals $S_{i}^{g} \in \S^{g}$. This setup is general enough to cover a broad range of signal designs from the information provision literature. For instance, $\S^{g} \subseteq \R$ corresponds to \textit{quantitative} signals, which includes government/institutional statistics \citep{cullen2022much, roth2022risk}, expert forecasts \citep{roth2020expectations, kumar2023effect}, and machine learning predictions \citep{jager2022worker, deshpande2023lack}. On the other hand, if we take $\S^{g}$ to be an abstract space, then we can accommodate \textit{qualitative} signals, such as informational videos \citep{alesina2021perceptions, dechezlepretre2022fighting}. 

\subsubsection{Action Function Primitives}\label{main:sec:primitive.conditions}

Define $B \mapsto \eta_{k}(B) = (\phi_{k'}(B))_{k'\neq k}$. In a slight abuse of notation, we let $\eta_{0}(B) := (\phi_{k}(B))_{k}$ so that $\eta_{0}(\B) := \{\eta_{0}(B): B \in \B\} = \{(\phi_{1}(B), \ldots, \phi_{K}(B)): B \in \B\}$. 

\begin{assumption}[Feature Action Functions]\label{main:ass:general.parametric}
$\eta_{0}(B)$ is convex, and there exists a continuously differentiable $h \mapsto Y_{i}^{*}(h)$ defined over it such that $B \mapsto Y_{i}(B) = Y_{i}^{*}(\eta_{0}(B))$.
\end{assumption}

To give primitive conditions for this assumption, we consider a (net) utility/profit function $u_{i}(\omega, y)$ that depends on states $\omega$ and actions $y \in \Y$. We can motivate the action function $B \mapsto Y_{i}(B)$ as agent $i$'s utility-maximizing map from beliefs to actions:
\begin{align}\label{main:eq:optimization.general}
    Y_{i}(B) = \argmax_{y \in \Y} \int u_{i}(\omega, y)  dB(\omega), \quad B \in \B. 
\end{align}
If agent $i$ has beliefs $B$, then $\int u_{i}(\omega, y) dB(\omega)$ is their subjective expected utility from taking action $y$. In particular, $Y_{i} \equiv Y_{i}(B_{i1})$ is the utility-maximizing action taken under $i$'s posterior beliefs $B_{i1}$. We can use representation (\ref{main:eq:optimization.general}) to construct feature action functions $Y_{i}^{*}(h)$. 

\paragraph{Restrictions on Preferences.} Let $m_{i}(\Y) := \{m_{i}(y): y \in \Y\}$ denote the image of $\Y$ under function $m_{i}: \Y \to \R$.

\begin{proposition}\label{main:prop:random.coefficient}
Consider $\phi_{k}(B) := \int \varphi_{k}(\omega) dB(\omega)$, where $\varphi_{k}: \Omega \to \R$. Let $\B$ be the subset of beliefs such that $\int |\varphi_{k}(\omega)|^{2}dB(\omega) < \infty$ for each $k$. If $B \mapsto Y_{i}(B)$ satisfies representation (\ref{main:eq:optimization.general}) with $u_{i}(\omega, y) = -(\sum_{k=1}^{K}\theta_{k i}\varphi_{k}(\omega) - m_{i}(y))^{2}$, where (i) $y \mapsto m_{i}(y)$ is strictly monotonic and continuously differentiable; and (ii) $\sum_{k = 1}^{K}\theta_{k i}\phi_{k}(B) \in m_{i}(\Y)$ for each $B \in \B$, then $Y_{i}(B) = m_{i}^{-1}(\sum_{k = 1}^{K}\theta_{k i}\phi_{k}(B))$. In particular, Assumption \ref{main:ass:general.parametric} is satisfied for $Y_{i}^{*}(h) := m_{i}^{-1}(\sum_{k = 1}^{K}\theta_{k i}h_{k})$.
\end{proposition}

Proposition \ref{main:prop:random.coefficient} specifies a class of squared-error loss functions $-u_{i}(\omega, y)$. To give intuition, we consider the example of \cite{jager2022worker}. Let $m_{i}(y) = y$, and suppose there is a latent function $\omega \mapsto f_{i}(\omega)$ that maps potential outside option wages $\omega$ to the optimal likelihood that worker $i$ should look for a new job. However, this function is complicated or imperfectly known, and so $i$ uses the approximation $\omega \mapsto \sum_{k=1}^{K}\theta_{k i}\varphi_{k}(\omega)$, where $\varphi_{k}(\omega) = \omega^{k-1}$. In this case, $Y_{i}(B) = \sum_{k=1}^{K}\theta_{ki}\phi_{k}(B)$ gives worker $i$'s minimum expected squared-error prediction for the approximate optimal likelihood that they should look for a new job. Notice that $K=2$ gives a linear approximation, whereas $K=3$ gives a quadratic approximation. In the former case, we obtain $Y_{i}(B) = \theta_{1i} + \theta_{2i}\mu(B)$, and $\theta_{2i}$ gives the partial effects of expectations on actions---this specification is considered in \cite{balla2023identifying}. On the other hand, $K=3$ allows agents to be risk averse, which is relevant for some applications \citep{kumar2023effect, coibion2021effect}.

When $m_{i}(y) = y$, the partial effects of feature $\phi_{k}$ above are homogeneous across $\eta_{0}(B)$:
\begin{align*}
    v \mapsto \partial_{v} Y_{i}^{*}(v, h_{-k}) = \theta_{k i}.
\end{align*}
However, this homogeneity may not hold in some models \citep{cullen2022much}. To generate nonlinear feature action functions, we can consider nonlinear $m_{i}$:
\begin{align*}
    v \mapsto \partial_{v} Y_{i}^{*}(v, h_{-k}) = (\partial_{y}m_{i}(Y_{i}^{*}(v, h_{-k})))^{-1}\theta_{ki}.
\end{align*}

\paragraph{Restrictions on Beliefs.} In some cases, $\B$ is a parametric family of beliefs. A leading class of parametric $\B$ are exponential families \citep[Section 3.4]{lehmann2006theory}. This class includes the Gaussian families often considered in models from the information provision literature \citep{armantier2016price, cavallo2017inflation, armona2019home, cullen2022much, fuster2022expectations, balla2022determinants}. There are also classes of parametric beliefs that are useful for belief elicitation; for example, \cite{kumar2023effect} and \cite{coibion2021effect} consider triangular distributions.

If $\B$ is a parametric family of beliefs characterized by $(\phi_{k})_{k=1}^{K}$, then $\eta_{0}$ is one-to-one over $\eta_{0}(\B)$ so that $Y_{i}^{*}(h) := Y_{i}(\eta_{0}^{-1}(h))$. Consider a sigma-finite measure $\nu \in \Delta(\Omega)$.

\begin{proposition}\label{main:prop:parametric.optimization}
Consider a set of beliefs $\B$ dominated by $\nu$, with densities $b(\omega) := dB(\omega)/d\nu$ parameterized by $\eta_{0}$ in the sense that there exists $\omega \mapsto b^{*}(\omega|h)$ with open and convex parameter space $\eta_{0}(\B)$ such that $b(\omega) = b^{*}(\omega|\eta_{0}(B))$ for each $B \in \B$. If $h \mapsto b^{*}(\omega|h)$ is continuously differentiable over $\eta_{0}(\B)$ almost surely-$\nu$ and $B \mapsto Y_{i}(B)$ satisfies representation (\ref{main:eq:optimization.general}) with $y \mapsto u_{i}(\omega, y)$ twice continuously differentiable and strictly concave, then under regularity conditions, Assumption \ref{main:ass:general.parametric} is satisfied for $Y_{i}^{*}(h) := Y(\eta_{0}^{-1}(h))$.
\end{proposition}

In general, Proposition \ref{main:prop:parametric.optimization} generates partial effects that are heterogeneous across $\eta_{0}(\B)$:
\begin{align*}
    v \mapsto \partial_{v} Y_{i}^{*}(v, h_{-k}) = \frac{\partial_{v}\partial_{y}U_{i}(Y_{i}^{*}(v, h_{-k}), v, h_{-k})}{|\partial_{y}^{2}U_{i}(Y_{i}^{*}(v, h_{-k}), v, h_{-k})|},
\end{align*}
where $U_{i}(y, h) := \int u_{i}(\omega, y) b^{*}(\omega|h) d\omega$. For intuition, consider agent $i$'s optimal behavior at feature value $h \in \eta_{0}(\B)$---or equivalently, at belief $\eta_{0}^{-1}(h) \in \B$. At these feature values, the optimal action is $Y_{i}^{*}(h)$. The marginal change in beliefs along feature $\phi_{k}$ leads to a marginal change $\partial_{v}U_{i}(Y_{i}^{*}(v, h_{-k}), v, h_{-k})$ in the agent's subjective expected utility function, so $Y_{i}^{*}(h)$ is no longer optimal. To re-optimize, the agent can either increase or decrease the intensity of their action. This choice of \textit{direction} depends on whether the associated marginal expected utility $\partial_{v}\partial_{y}U_{i}(Y_{i}^{*}(v, h_{-k}), v, h_{-k})$ is positive or negative at $Y_{i}^{*}(v, h_{-k})$. The \textit{magnitude} of the change in actions depends on the curvature $|\partial_{y}^{2}U_{i}(Y_{i}^{*}(v, h_{-k}), v, h_{-k})|$ of the expected utility function at $Y_{i}^{*}(v, h_{-k})$.

\subsubsection{Identifying Variation}
Given $\phi_{k}$ and a pair of groups $g, \tilde{g} \in \G$, we consider parameters of the form
\begin{align*}
    \beta_{k}^{g\Tilde{g}} := \E[w_{ki}^{g\Tilde{g}}\Bar{\beta}_{ki}^{g\Tilde{g}}], \quad \quad w_{ki}^{g\Tilde{g}} &:= \frac{|\phi_{k}(B_{i1}^{\Tilde{g}}) - \phi_{k}(B_{i1}^{g})|}{\E[|\phi_{k}(B_{i1}^{\Tilde{g}}) - \phi_{k}(B_{i1}^{g})|]}, \\[10pt]
     \Bar{\beta}_{ki}^{g\Tilde{g}} &:= \int \lambda_{ki}^{g\Tilde{g}}(v)\partial_{v}Y_{i}^{*}(v, \eta_{k}(B_{i1}^{g}))dv,
\end{align*}
where $v \mapsto \lambda_{ki}^{g\Tilde{g}}(v)$ is the density function of the uniform distribution over the feature values between $\phi_{k}(B_{i1}^{g})$ and $\phi_{k}(B_{i1}^{\Tilde{g}})$. Let $\psi_{ki}^{g\Tilde{g}} := \text{sign}(\phi_{k}(B_{i1}^{\Tilde{g}}) - \phi_{k}(B_{i1}^{g}))$. To recover $\beta_{k}^{g\Tilde{g}}$, we make the following assumption.

\begin{assumption}[Identifying Variation]\label{main:ass:general.variation}
For each $\phi_{k} \in \Phi \cap \{\phi_{k}\}_{k}$, there exist $g, \tilde{g} \in \G$ such that (i) $\eta_{k}(B_{i1}^{\Tilde{g}}) = \eta_{k}(B_{i1}^{g})$; and (ii)
\begin{align*}
    c_{ki}^{g\Tilde{g}}\psi_{ki}^{g\Tilde{g}} = |\psi_{ki}^{g\Tilde{g}}|, \quad c_{ki}^{g\Tilde{g}} := c_{ki}^{g\Tilde{g}}(X_{i}, (S_{i}^{g})_{g \in \G}, (\phi(B_{i0}))_{\phi \in \Phi}).
\end{align*}
for a known function $c_{ki}^{g\Tilde{g}}(\cdot)$ defined over the support of $(X_{i}, (S_{i}^{g})_{g \in \G}, (\phi(B_{i0}))_{\phi \in \Phi})$.
\end{assumption}
Assumption \ref{main:ass:general.variation}(i) requires the experiment to generate ceteris paribus variation in feature $\phi_{k}$. Assumption \ref{main:ass:general.variation}(ii) requires the existence of correction terms $c_{ki}^{g\Tilde{g}}$ for the signs $\psi_{ki}^{g\Tilde{g}}$. In Specification \ref{main:spec:active} for active control, this corresponds to $\text{sign}(S_{i}^{T} - S_{i}^{L}) = 1$, which in that case is constant. In Specification \ref{main:spec:passive} for passive control Interaction \ref{main:int:IPIV}, this corresponds to $\text{sign}(S_{i}^{T} - \phi(B_{i0}))$. However, Assumption \ref{main:ass:general.variation}(ii) allows the function $c_{ki}^{g\Tilde{g}}(\cdot)$ to depend on the full set of observed variables that are independent of group assignment $G_{i}$. The potential flexibility of $c_{ki}^{g\Tilde{g}}(\cdot)$ allows the experiment to use a multitude of institutional or theoretical arguments to construct correction terms. In particular, this formulation covers experiments with qualitative information.

\subsubsection{TSLS Estimation}\label{main:sec:general.estimation}

\begin{specification}[Conditional]\label{main:spec:conditional}
Let $W_{i}$ be a covariate vector that includes $c_{ki}^{g\Tilde{g}}$ and potentially other observed variables independent of $G_{i}$. Let $F_{i} = W_{i}'\pi_{0} + \pi\1\{G_{i} = \tilde{g}\}c_{ki}^{g\Tilde{g}}$ be the first-stage population linear regression of $\phi_{k}(B_{i1})$ on $W_{i}$ and $\1\{G_{i} = \tilde{g}\}c_{ki}^{g\Tilde{g}}$ conditional on $G_{i} \in \{g, \Tilde{g}\}$. The TSLS specification of interest is 
\begin{align*}
    \phi_{k}(B_{i1}) &= W_{i}'\pi_{0} + \pi\1\{G_{i} = \tilde{g}\}c_{ki}^{g\Tilde{g}} + \zeta_{i}, \\
    Y_{i} &= W_{i}'\gamma_{0} + \gamma F_{i} + \upsilon_{i},
\end{align*}
where $\zeta_{i}$ and $\upsilon_{i}$ are population regression residuals. The TSLS coefficient of interest is $\gamma$.
\end{specification}

In what follows, let $\E_{g\Tilde{g}}[\cdot] := \E[\cdot|G_{i} \in \{g, \Tilde{g}\}]$.

\begin{proposition}\label{main:prop:conditional}
Suppose that Assumptions \ref{main:ass:IV}, \ref{main:ass:general.parametric}, and \ref{main:ass:general.variation} are satisfied. Consider Specification \ref{main:spec:conditional} for $\phi_{k} \in \Phi \cap \{\phi_{k}\}$ and an associated pair of groups $g, \tilde{g} \in \G$. If (i) $\E_{g\Tilde{g}}[W_{i}W_{i}']$ is full-rank and $\P(G_{i} = \Tilde{g}) \in (0,1)$; and (ii) $\E_{g\Tilde{g}}[c_{ki}^{g\Tilde{g}}\phi_{k}(B_{i1})|G_{i} = \Tilde{g}] - \E_{g\Tilde{g}}[c_{ki}^{g\Tilde{g}}\phi_{k}(B_{i1})|G_{i} = g] \neq 0$, then Specification \ref{main:spec:conditional} identifies $\gamma = \beta_{k}^{g\Tilde{g}}$.
\end{proposition}

If there are multiple pairs $\{g, \Tilde{g}\}$ such that $\beta_{k}^{g\Tilde{g}}$ is identified for a fixed $k$, then we can also aggregate these parameters. For example, we can take a control group $g$ (i.e., $\S^{g} = \{\varnothing\}$) as the baseline comparison group, and then design the treatment groups $\Tilde{g} \neq g$ to shift expectations $\mu$---this speaks to, for instance, experiments that consider the impact of different forms of communication \citep{coibion2022monetary}. In particular, given choice of weights $\alpha^{\Tilde{g}}_{k}$ that are non-negative and sum to one, we can construct
\begin{align*}
    \sum_{\Tilde{g} \neq g}\alpha^{\Tilde{g}}_{k}\beta_{k}^{g\Tilde{g}}.
\end{align*}
We advise against using standard TSLS specifications for such estimands. Even if there were no contamination issues, the presence of multiple group indicators would create a multiple instrument problem \citep{mogstad2021causal}. Relatedly, TSLS specifications with multiple features as endogenous variables generally do not produce interpretable aggregations \citep{bhuller20222sls}.

\subsection{Lemmas}
\paragraph{Partial Effect Representation.}
For a given feature $\phi: \B \to \R$, consider some function $v \mapsto f_{i}(v)$ that is defined over $\{\phi(B): B \in \B\}$. Let $\V_{i}^{g\Tilde{g}}$ be the set of feature values between $\phi(B_{i1}^{g})$ and $\phi(B_{i1}^{\tilde{g}})$, and $\psi_{i}^{g\Tilde{g}} := \text{sign}(\phi(B_{i1}^{\tilde{g}}) - \phi(B_{i1}^{g}))$.
\begin{lemma}\label{main:lemma:FTC}
If $\{\phi(B): B \in \B\}$ is convex and $v \mapsto f_{i}(v)$ is continuously differentiable, then
\begin{align*}
    f_{i}(\phi(B_{i1}^{\tilde{g}})) - f_{i}(\phi(B_{i1}^{g})) &= \psi_{i}^{g\Tilde{g}} \int \partial_{v}f_{i}(v) \1\{v \in \V_{i}^{g\Tilde{g}}\} dv.
\end{align*}
In particular, if $|\phi(B_{i1}^{\tilde{g}}) - \phi(B_{i1}^{g})| \neq 0$, then
\begin{align*}
    f_{i}(\phi(B_{i1}^{\tilde{g}})) - f_{i}(\phi(B_{i1}^{g})) = \psi_{i}^{g\Tilde{g}}|\phi(B_{i1}^{\tilde{g}}) - \phi(B_{i1}^{g})| \int \lambda_{i}^{g\tilde{g}}(v)\partial_{v}f_{i}(v) dv,
\end{align*}
where 
\begin{align*}
    \lambda_{i}^{g\tilde{g}}(v) := \frac{\1\{v \in \V_{i}^{g\Tilde{g}}\}}{\displaystyle \int \1\{v \in \V_{i}^{g\Tilde{g}}\} dv}
\end{align*}
is the density function corresponding to the uniform distribution over $\V_{i}^{g\Tilde{g}}$.
\end{lemma}
\begin{proof}[Proof of Lemma \ref{main:lemma:FTC}]
$ $ \newline
This follows from the fundamental theorem of calculus. 
\end{proof}

\paragraph{General Specification.} The considered TSLS specifications are special cases of the following specification. 
\begin{specification}[General]\label{main:spec:fully.general}
Let $I_{i}$ be an observed interaction vector such that $G_{i} \ind I_{i}$. Let $W_{i}$ be a covariate vector that includes $I_{i}$ and potentially other observed variables independent of $G_{i}$. Let $F_{i} = W_{i}'\pi_{0} + \1\{G_{i} = \tilde{g}\}I_{i}'\pi$ be the first-stage population linear regression of $\phi(B_{i1})$ on $W_{i}$ and $\1\{G_{i} = \tilde{g}\}I_{i}$. The TSLS specification of interest is 
\begin{align*}
    \phi(B_{i1}) &= W_{i}'\pi_{0} + \1\{G_{i} = \Tilde{g}\}I_{i}'\pi + \zeta_{i}, \\
    Y_{i} &= W_{i}'\gamma_{0} + \gamma F_{i} + \upsilon_{i},
\end{align*}
where $\zeta_{i}$ and $\upsilon_{i}$ are population regression residuals. The coefficient of interest is $\gamma$.
\end{specification}

\begin{lemma}\label{main:lemma:fully.general}
Consider $\G = \{g, \Tilde{g}\}$ and suppose that Assumption \ref{main:ass:IV} is satisfied. If (i) $\E[W_{i}W_{i}']$ is full-rank and $\P(G_{i} = \Tilde{g}) \in (0,1)$; and (ii) $\E[I_{i}\phi(B_{i1})|G_{i} = \Tilde{g}] - \E[I_{i}\phi(B_{i1})|G_{i} = g] \neq 0$, then Specification \ref{main:spec:fully.general} identifies  
\begin{align*}
    \gamma = \frac{\E[I_{i}'\pi(Y_{i}(B_{i1}^{\tilde{g}})- Y_{i}(B_{i1}^{g}))]}{\E[I_{i}'\pi(\phi(B_{i1}^{\tilde{g}})- \phi(B_{i1}^{g}))]}.
\end{align*}
\end{lemma}

\begin{proof}[Proof of Lemma \ref{main:lemma:fully.general}]
$ $ \newline
The rank condition on $\E[W_{i}W_{i}']$ means that the coefficients from a given linear regression on $W_{i}$ are identified. Assumption \ref{main:ass:IV} implies that the linear regression of $\1\{G_{i} = \Tilde{g}\}I_{i}$ on $W_{i}$ is $\P(G_{i} = \Tilde{g})I_{i}$, since $I_{i}$ is included in $W_{i}$. Moreover $\P(G_{i} = \Tilde{g}) \in (0,1)$ implies that the second moment matrix of the residuals $(\1\{G_{i} = \Tilde{g}\} - \P(G_{i} = \Tilde{g}))I_{i}$ is full-rank. In particular, the first-stage coefficients are identified. Letting $\Tilde{F}_{i} := F_{i} - \LP[F_{i}|W_{i}]$, where $\LP[F_{i}|W_{i}]$ is the linear regression of $F_{i}$ on $W_{i}$, the orthogonality conditions imply
\begin{align*}
    \E[\Tilde{F}_{i}Y_{i}] = \gamma\E[\Tilde{F}_{i}F_{i}] = \gamma\E[\Tilde{F}_{i}\phi(B_{i1})].
\end{align*}
$F_{i} = W_{i}'\pi_{0} + \1\{G_{i} = \tilde{g}\}I_{i}'\pi$ implies $\LP[F_{i}|W_{i}] = W_{i}'\pi_{0} + \P(G_{i} = \Tilde{g})I_{i}'\pi$. Therefore,
\begin{align*}
     \E[\Tilde{F}_{i}Y_{i}] &= \E[(\1\{G_{i} = \Tilde{g}\} - \P(G_{i} = \Tilde{g}))Y_{i}I_{i}'\pi] \\
     &= \var(\1\{G_{i} = \Tilde{g}\})(\E[Y_{i}I_{i}'\pi|G_{i}  = \Tilde{g}] - \E[Y_{i}I_{i}'\pi|G_{i} = g]) \\
     &= \var(\1\{G_{i} = \Tilde{g}\})\E[I_{i}'\pi (Y_{i}(B_{i1}^{\Tilde{g}}) - Y_{i}(B_{i1}^{g}))],
\end{align*}
where the last equality uses Assumption \ref{main:ass:IV}. We can proceed analogously with 
\begin{align*}
    \E[\Tilde{F}_{i}\phi(B_{i1})] \propto \E[I_{i}'\pi(\phi(B_{i1}^{\tilde{g}})- \phi(B_{i1}^{g}))] &= \pi'\E[I_{i}(\phi(B_{i1}^{\tilde{g}})- \phi(B_{i1}^{g}))] \\
    &= \pi'(\E[I_{i}\phi(B_{i1})|G_{i} = \Tilde{g}]- \E[I_{i}\phi(B_{i1})|G_{i} = g]) \\
    &= \pi'\E[I_{i}I_{i}']\pi.
\end{align*}
$\E[I_{i}\phi(B_{i1})|G_{i} = \Tilde{g}]- \E[I_{i}\phi(B_{i1})|G_{i} = g] \neq 0$, and so $\pi \neq 0$. In particular, $\E[\Tilde{F}_{i}\phi(B_{i1})] > 0$, and so $\gamma$ is identified from  
\begin{align*}
    \E[I_{i}'\pi(Y_{i}(B_{i1}^{\tilde{g}})- Y_{i}(B_{i1}^{g}))] &= \gamma\E[I_{i}'\pi(\phi(B_{i1}^{\tilde{g}})- \phi(B_{i1}^{g}))].
\end{align*}
\end{proof}

\subsection{Proofs}\label{main:sec:proofs}

\begin{proof}[Proof of Proposition \ref{main:prop:passive.baseline}]
$ $ \newline
The conditions of Lemma \ref{main:lemma:fully.general} are satisfied with $\G = \{C, T\}$:
\begin{align*}
    \gamma = \frac{\E[I_{i}'\pi (Y_{i}(B_{i1}^{T})- Y_{i}(B_{i1}^{C}))]}{\E[I_{i}'\pi(\phi(B_{i1}^{T})- \phi(B_{i1}^{C}))]} &= \frac{\E[I_{i}'\pi\psi_{i}^{CT}|\phi(B_{i1}^{T}) - \phi(B_{i1}^{C})| \Bar{\beta}_{i}^{CT}]}{\E[I_{i}'\pi\psi_{i}^{CT}|\phi(B_{i1}^{T}) - \phi(B_{i1}^{C})|]},
\end{align*}
where the second equality follows from Assumption \ref{main:ass:parametric} and Lemma \ref{main:lemma:FTC}.
\end{proof}

\begin{proof}[Proof of Proposition \ref{main:prop:active.baseline}]
$ $ \newline
The conditions of Lemma \ref{main:lemma:fully.general} are satisfied with $\G = \{L,H\}$ and $I_{i} = 1$:
\begin{align*}
    \gamma = \frac{\E[(Y_{i}(B_{i1}^{H})- Y_{i}(B_{i1}^{L}))]}{\E[(\phi(B_{i1}^{H})- \phi(B_{i1}^{L}))]} &= \frac{\E[\psi_{i}^{LH}|\phi(B_{i1}^{H}) - \phi(B_{i1}^{L})| \Bar{\beta}_{i}^{LH}]}{\E[\psi_{i}^{LH}|\phi(B_{i1}^{H}) - \phi(B_{i1}^{L})|]},
\end{align*}
where the second equality follows from Assumption \ref{main:ass:parametric} and Lemma \ref{main:lemma:FTC}.
\end{proof}

\begin{proof}[Proof of Proposition \ref{main:prop:passive.MLRP}]
$ $ \newline
Assumption \ref{main:ass:signal.monotonicity}, stability, and $S_{i}^{T} \neq \phi(B_{i0})$ imply $\text{sign}(S_{i}^{T} - \phi(B_{i0}))\psi_{i}^{CT} = |\psi_{i}^{CT}|$. 
\end{proof}

\begin{proof}[Proof of Proposition \ref{main:prop:active.MLRP}]
$ $ \newline
Assumption \ref{main:ass:signal.monotonicity} and neutrality imply $\psi_{i}^{LH} = \text{sign}(S_{i}^{H} - S_{i}^{L}) = 1$.
\end{proof}

\begin{proof}[Proof of Proposition \ref{main:prop:MLRP}]
$ $ \newline
If $\omega' > \omega$, then $s \mapsto q_{i1}^{g}(s|\omega')/q_{i1}^{g}(s|\omega)$ is increasing function in $s \in \R$. In particular, if $B_{i1}^{g}(\omega|s)$ is a Bayesian rule, then
\begin{align*}
    \frac{dB_{i1}^{g}(\omega'|s)/d\nu}{dB_{i1}^{g}(\omega|s)/d\nu} = \frac{q_{i1}^{g}(s|\omega')}{q_{i1}^{g}(s|\omega)}\frac{b_{i0}(\omega')}{b_{i0}(\omega)}
\end{align*}
is increasing in $s$, which means $s \mapsto B_{i1}^{g}(\cdot|s)$ is increasing in the sense of first-order stochastic dominance. Therefore, since $\varphi(\omega)$ is increasing in $\omega$, then $\phi(B_{i1}^{g}(\omega|s))$ is increasing in $s \in \R$. 

If $B_{i1}^{g}(\omega|s)$ is instead anchored, then the form of $\phi$ implies
\begin{align*}
    \phi(B_{i1}^{g}(\omega|s)) = \tau_{i}^{g}\phi(\Bar{B}_{i}^{g}) + (1 - \tau_{i}^{g}) \int \varphi(\omega) \frac{q_{i1}^{g}(s|\omega)b_{i0}(\omega)}{\displaystyle \int q_{i1}^{g}(s|\omega)b_{i0}(\omega) d\nu(\omega)} d\nu(\omega).  
\end{align*}
The latter is the attenuated value of $\phi$ evaluated at a Bayesian rule, and so is increasing in $s \in \R$. This implies that $\phi(B_{i1}^{g}(\omega|s))$ is increasing in $s \in \R$.

If $B_{i1}^{g}(\omega|s)$ is instead Grether, then
\begin{align*}
    \frac{dB_{i1}^{g}(\omega'|s)/d\nu}{dB_{i1}^{g}(\omega|s)/d\nu} = \left(\frac{q_{i1}^{g}(s|\omega')}{q_{i1}^{g}(s|\omega)}\right)^{\chi_{i1}^{g}}\left(\frac{b_{i0}(\omega')}{b_{i0}(\omega)}\right)^{^{\chi_{i0}^{g}}},
\end{align*}
which is increasing in $s \in \R$, which means $s \mapsto B_{i1}^{g}(\cdot|s)$ is increasing in the sense of first-order stochastic dominance. Therefore, signal monotonicity follows as in the Bayesian case. 
\end{proof}

\begin{proof}[Proof of Proposition \ref{main:prop:random.coefficient}]
$ $ \newline
Consider the problem of choosing $a \in m_{i}(\Y)$ to minimize $\int (\sum_{k=1}^{K}\theta_{k i}\varphi_{k}(\omega) - a)^{2}dB(\omega)$. This is solved by $a_{i}(B) = \sum_{k=1}^{K}\theta_{k i}\phi_{k}(B)$. Note that the inverse of $m_{i}$ exists by strict monotonicity. Therefore, for each $y \neq m^{-1}(a_{i}(B))$, we have
\begin{align*}
    \int u_{i}(\omega, m_{i}^{-1}(a_{i}(B))) dB(\omega) > \displaystyle \int u_{i}(\omega, y) dB(\omega).
\end{align*}
Thus, $Y_{i}(B) = m_{i}^{-1}(\sum_{k = 1}^{K}\theta_{k i}\phi_{k}(B))$. The inverse function theorem implies that $h \mapsto Y_{i}^{*}(h)$ is continuously differentiable over $\R^{K}$ with derivatives $\partial_{v} Y_{i}^{*}(v, h_{-k}) = (\partial_{y}m_{i}(Y_{i}^{*}(v, h_{-k}))^{-1}\theta_{ki}$. To show convexity of $\eta_{0}(\B)$, consider $h, h' \in \eta_{0}(\B)$. By definition, there exist $B_{h}, B_{h'} \in \B$ such that $h_{k} = \phi_{k}(B_{h})$ and $h_{k}' = \phi_{k}(B_{h'})$ for each $k$. The definitions of $\B$ and $\eta_{0} := (\phi_{k})_{k=1}^{K}$ mean that $(1 - \alpha)B_{h} + \alpha B_{h'} \in \B$. Therefore, $(1 - \alpha)h + \alpha h' \in \eta(\B)$, and so Assumption \ref{main:ass:general.parametric} is satisfied.
\end{proof}

\begin{proof}[Proof of Proposition \ref{main:prop:parametric.optimization}]
$ $ \newline
To appeal to the dominated convergence theorem, assume that for each $y \in \Y$ and $h \in \eta_{0}(\B)$,
\begin{enumerate}[label=(\roman*)]
    \item $\int |u_{i}(\omega, y)|b^{*}(\omega|h)d\omega < \infty$;
    \item there exists $\delta > 0$ such that $\int \sup_{\Tilde{y}:|\tilde{y} - y| \leq \delta}|\partial_{y}^{n}u_{i}(\omega, \tilde{y})| b^{*}(\omega|h) d\omega < \infty$ for each $n \in \{1,2\}$;
    \item there exists $\delta > 0$ such that $\int \sup_{\Tilde{v}: |\Tilde{v} - v| \leq \delta}|\partial_{y}u_{i}(\omega, y)\partial_{v}b^{*}(\omega|\Tilde{v}, h_{-k})|d\omega < \infty$ for each $k$.
\end{enumerate}
These conditions allow us to exchange limits and integrals. Let $U_{i}(y, h) := \int u_{i}(\omega, y) b^{*}(\omega|h) d\omega$. Observe that $\partial_{y}^{2}U_{i}(y, h) = \int \partial_{y}^{2}u_{i}(\omega, y) b^{*}(\omega|h) d\omega < 0$. Thus, since $Y_{i}(B)$ satisfies representation (\ref{main:eq:optimization.general}), we have $\partial_{y}U_{i}(Y_{i}^{*}(h), h) = \int \partial_{y}u_{i}(\omega, Y_{i}^{*}(h)) b^{*}(\omega|h) d\omega = 0$. By the implicit function theorem,
\begin{align*}
    \partial_{v} Y_{i}^{*}(v, h_{-k}) = \frac{\partial_{v}\partial_{y}U_{i}(Y_{i}^{*}(v, h_{-k}), v, h_{-k})}{|\partial_{y}^{2}U_{i}(Y_{i}^{*}(v, h_{-k}), v, h_{-k})|},
\end{align*}
which is continuous over the convex parameter space $\eta(\B)$. Thus, Assumption \ref{main:ass:general.parametric} is satisfied.
\end{proof}

\begin{proof}[Proof of Proposition \ref{main:prop:conditional}]
$ $ \newline
The conditions of Lemma \ref{main:lemma:fully.general} are satisfied conditional on $G_{i} \in \{g, \Tilde{g}\}$ with $I_{i} = c_{i}^{g\tilde{g}}$. Following the proof of Lemma \ref{main:lemma:fully.general}  yields
\begin{align*}
    \gamma = \frac{\E_{g\Tilde{g}}[c_{i}^{g\tilde{g}}Y_{i}|G_{i} = \Tilde{g}] - \E_{g\Tilde{g}}[c_{i}^{g\tilde{g}}Y_{i}|G_{i} = g]}{\E_{g\Tilde{g}}[c_{i}^{g\tilde{g}}\phi_{k}(B_{i1})|G_{i} = \Tilde{g}] - \E_{g\Tilde{g}}[c_{i}^{g\tilde{g}}\phi_{k}(B_{i1})|G_{i} = g]} &= \frac{\E[c_{i}^{g\tilde{g}}(Y_{i}(B_{i1}^{\tilde{g}})- Y_{i}(B_{i1}^{g}))]}{\E[c_{i}^{g\tilde{g}}(\phi_{k}(B_{i1}^{\tilde{g}})- \phi_{k}(B_{i1}^{g}))]} \\
    &= \frac{\E[|\phi_{k}(B_{i1}^{\tilde{g}})- \phi_{k}(B_{i1}^{g})|\Bar{\beta}_{i}^{g\Tilde{g}}]}{\E[|\phi_{k}(B_{i1}^{\tilde{g}})- \phi_{k}(B_{i1}^{g})|]},
\end{align*}
where the first equality follows from Assumption \ref{main:ass:IV} and the second from Assumption \ref{main:ass:general.parametric}, Lemma \ref{main:lemma:FTC}, and Assumption \ref{main:ass:general.variation}(ii).
\end{proof}

\subsection{Applications}

\subsubsection{Additional Estimated Coefficients}\label{main:sec:appendix.applications.1}

\renewcommand{\thefigure}{E\arabic{figure}}
\setcounter{figure}{0}

Figures \ref{fig:kumar_appendix_coefficients} and \ref{fig:jager_appendix_coefficients} display coefficients from the applications that are not shown in the main text.

\begin{figure}[ht!]
    \centering
    \includegraphics[width = \textwidth]{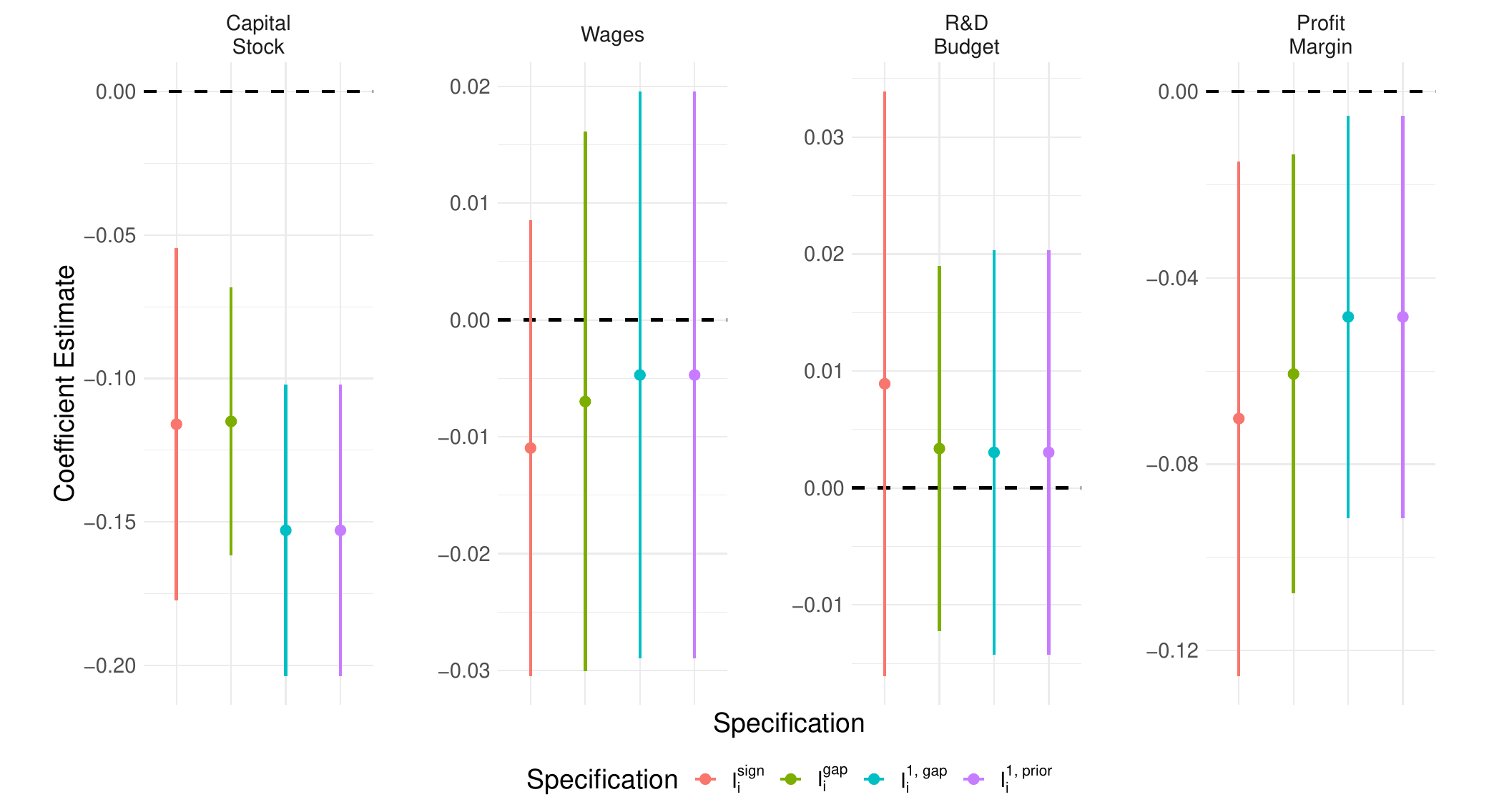}
    \caption{Estimated Coefficients---Data from \cite{kumar2023effect}}
    \label{fig:kumar_appendix_coefficients}
    \floatfoot{\textit{Note:} Outcomes are the all economic outcomes from \cite{kumar2023effect}, with the exception of price, employment, and advertising budget, which are in the main text. From left to right, the outcomes are the realized change in the capital stock of the firm, the realized change in the wages of the firm, the realized change in the R\&D budget of the firm, and the realized change in the profit margin of the firm, all relative to the planned change in these outcomes six months prior. The interactions and estimation procedures are described in Section \ref{main:sec:interpretation}. The points are point-estimates and the bars are 95\% confidence intervals.}
\end{figure}

\begin{figure}[h!]
    \centering
    \includegraphics[width = \textwidth]{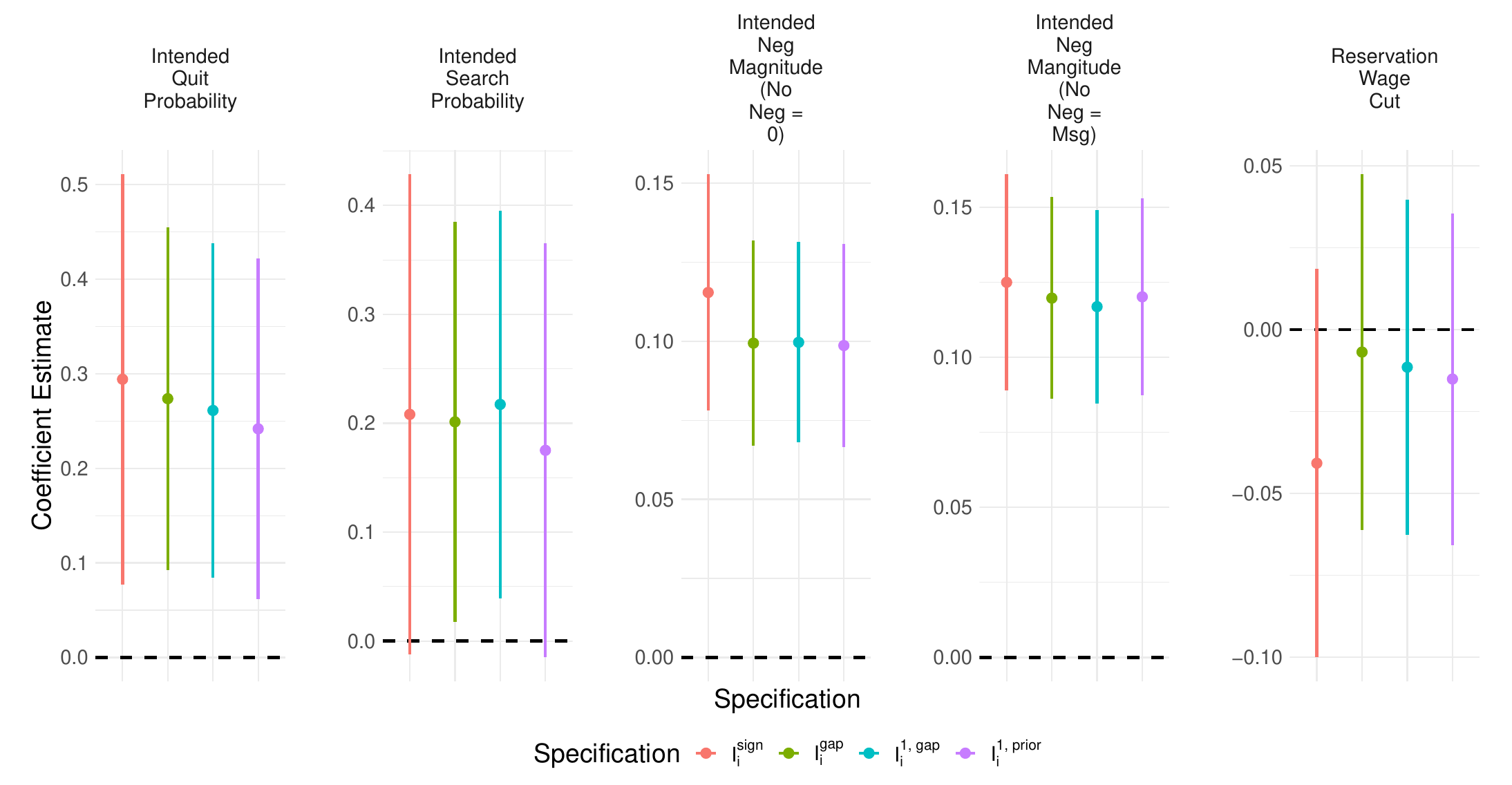}
    \caption{Estimated Coefficients---Data from \cite{jager2022worker}}
    \label{fig:jager_appendix_coefficients}
    \floatfoot{\textit{Note:} Outcomes are all economic outcomes from \cite{jager2022worker}, with the exception of intended negotiation probability, which is in the main text. From left to right, the outcomes are the worker's probability of quitting their current job, the probability of finding another job, the intended negotiation magnitude (no negotiations coded as 0), the intended negotiation magnitude (no negotiations coded as missing), and the reservation wage cut as a percent of their current wage. The interactions and estimation procedures are described in Section \ref{main:sec:interpretation}. The points are point-estimates and the bars are 95\% confidence intervals.}
\end{figure}

\subsubsection{Weight Characterization for Jager et al. (2024)}\label{main:sec:appendix.applications.2}

For the intended negotiation probability outcome, we characterize the average weight that TSLS gives to agents for each decile of the perception gap: $\E[w_{i}^{CT}(I)|(S_{i}^{T} - \mu(B_{i0})) \in [a_j, b_j]]$, where $a_j, b_j$ are the bounds of the deciles and $w_{i}^{CT}(I)$ denotes the weights on agent $i$ when the interaction is $I_{i}$. \label{main:sec:complier.characterization}
Given $X_{i} \ind G_{i}$ and Assumption \ref{main:ass:IV}, we have\begin{align}
    \E[w_{i}^{CT}(I)f(X_{i})] = \frac{\E[I_{i}'\pi f(X_{i})\mu(B_{i1})|G_{i} = T] - \E[I_{i}'\pi f(X_{i})\mu(B_{i1})|G_{i} = C]}{\E[I_{i}'\pi\mu(B_{i1})|G_{i} = T] - \E[I_{i}'\pi\mu(B_{i1})|G_{i} = C]},
\end{align}
where $f$ is some function of $X_{i}$. For example, if $X_{i} \in [x_{min}, x_{max}]$, then
\begin{align*}
    f_{a,b}(X_{i}) := \frac{\1\{X_{i} \in [a,b]\}}{\P(X_{i} \in [a,b])}, \quad x_{min} \leq a < b \leq x_{max},
\end{align*}
gives
\begin{align*}
    \E[w_{i}^{CT}(I)f_{a,b}(X_{i})] = \E[w_{i}^{CT}(I)|X_{i} \in [a,b]].
\end{align*}
In this example, if we have a collection of intervals $[a_{j}, b_{j}]$ such that $\cup_{j=1}^{J}[a_{j}, b_{j}] = [x_{min}, x_{max}]$, then ranging $\E[w_{i}^{CT}(I)f_{a_{j},b_{j}}(X_{i})]$ over $j = 1, \ldots, J$ allows us to compute the average weight of group $[a_{j}, b_{j}]$. We can also consider
\begin{align*}
    \beta^{CT}(I) := \E[w_{i}^{CT}(I)\Bar{\beta}_{i}^{CT}] = \sum_{j=1}^{J}\P(X_{i} \in [a_{j}, b_{j}])\E[w_{i}^{CT}(I)\Bar{\beta}_{i}^{CT}|X_{i} \in [a_{j}, b_{j}]],
\end{align*}
where each $\E[w_{i}^{CT}(I)\Bar{\beta}_{i}^{CT}|X_{i} \in [a_{j}, b_{j}]]$ is identified as
\begin{align*}
    \E[w_{i}^{CT}(I)\Bar{\beta}_{i}^{CT}|X_{i} \in [a_{j}, b_{j}]] = \frac{\E[I_{i}'\pi f_{a_{j}, b_{j}}(X_{i})Y_{i}|G_{i} = T] - \E[I_{i}'\pi f_{a_{j}, b_{j}}(X_{i})Y_{i}|G_{i} = C]}{\E[I_{i}'\pi\mu(B_{i1})|G_{i} = T] - \E[I_{i}'\pi\mu(B_{i1})|G_{i} = C]}.
\end{align*}

Figure \ref{fig:jager-weight} plots the weights against the perception gap. Compared to interactions $I_i^{gap}$, $I_i^{1,gap}$, and $I_i^{1,prior}$, $I_i^{sign}$ places relatively less weight on those in the lowest decile, and relatively more weight on those in the middle deciles. Therefore, if agents with moderate perception gaps have larger APEs, then the $I_i^{sign}$ coefficient will be larger than those for interactions $I_i^{gap}$, $I_i^{1,gap}$, and $I_i^{1,prior}$, which would rationalize the estimates in Figure \ref{fig:jager-coeff}.

\begin{figure}[h!]
    \centering
    \includegraphics[width = \textwidth]{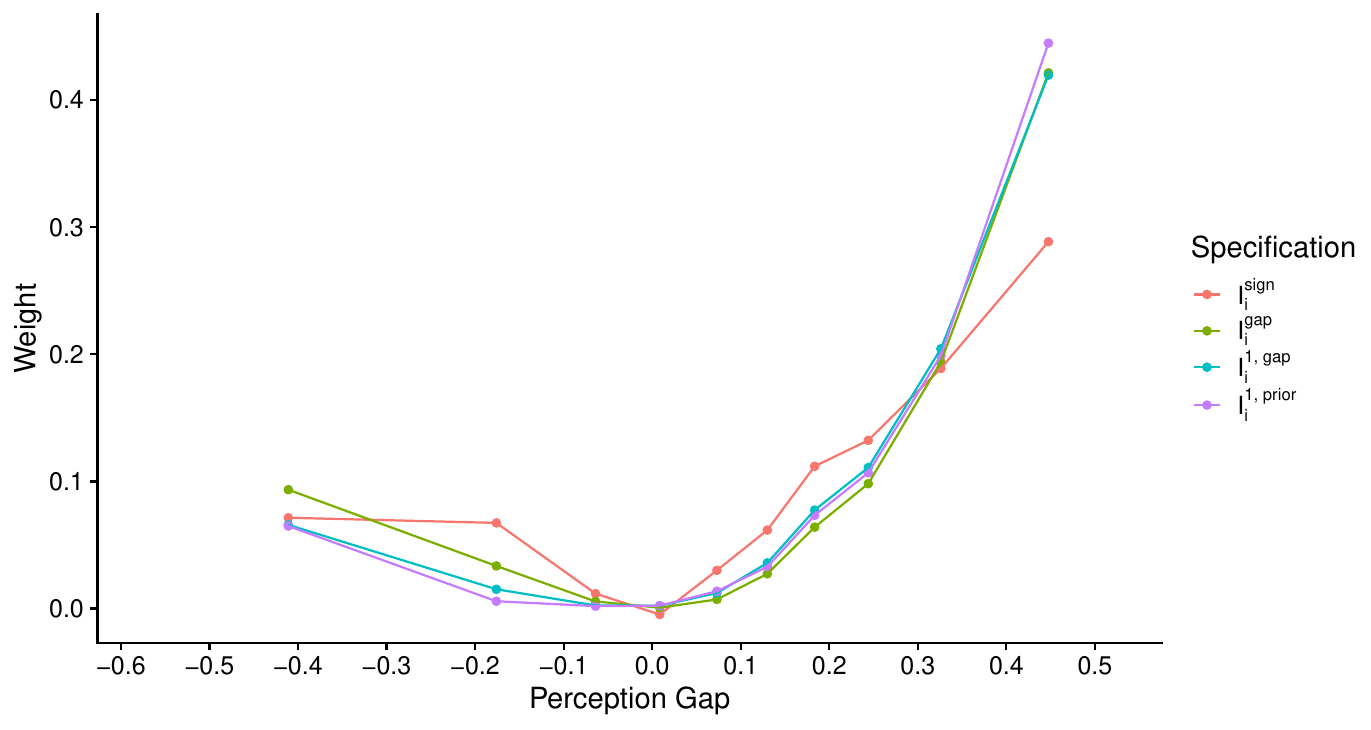}
    \caption{Perception Gap Characterization of Weights---Data from \cite{jager2022worker}}
    \floatfoot{\textit{Note:} Each point represents the characterized TSLS weight for the corresponding decile bin, using intended negotiation probability as the outcome. For example, the leftmost point corresponds to the characterized weight for those in the lowest decile bin, (-0.572,-0.228]. For a description of the characterization procedure, see Appendix \ref{main:sec:complier.characterization}
.}
    \label{fig:jager-weight}
\end{figure}

We can also characterize the contribution of each decile of perception gap to the estimated TSLS coefficients. Denote the contribution as $\E[w_{i}^{CT}(I)\Bar{\beta}_{i}^{CT}|(S_{i}^{T} - \mu(B_{i0})) \in [a_j, b_j]]$. Figure \ref{fig:jager_weight_Y} plots these contributions against the perception gaps. Comparing the shape of the plot in Figure \ref{fig:jager_weight_Y} against that of Figure \ref{fig:jager-weight} shows where the differences in the APEs are. For example, small differences in Figure \ref{fig:jager-weight} weight contributions for the leftmost deciles translate to large differences in the analogous Figure \ref{fig:jager_weight_Y} TSLS estimate contributions. This is evidence of heterogeneous partial effects across agents---if these effects were constant, then the plots in Figures \ref{fig:jager-weight} and \ref{fig:jager_weight_Y} would have the same shape.


\begin{figure}
    \centering
    \includegraphics[width = \textwidth]{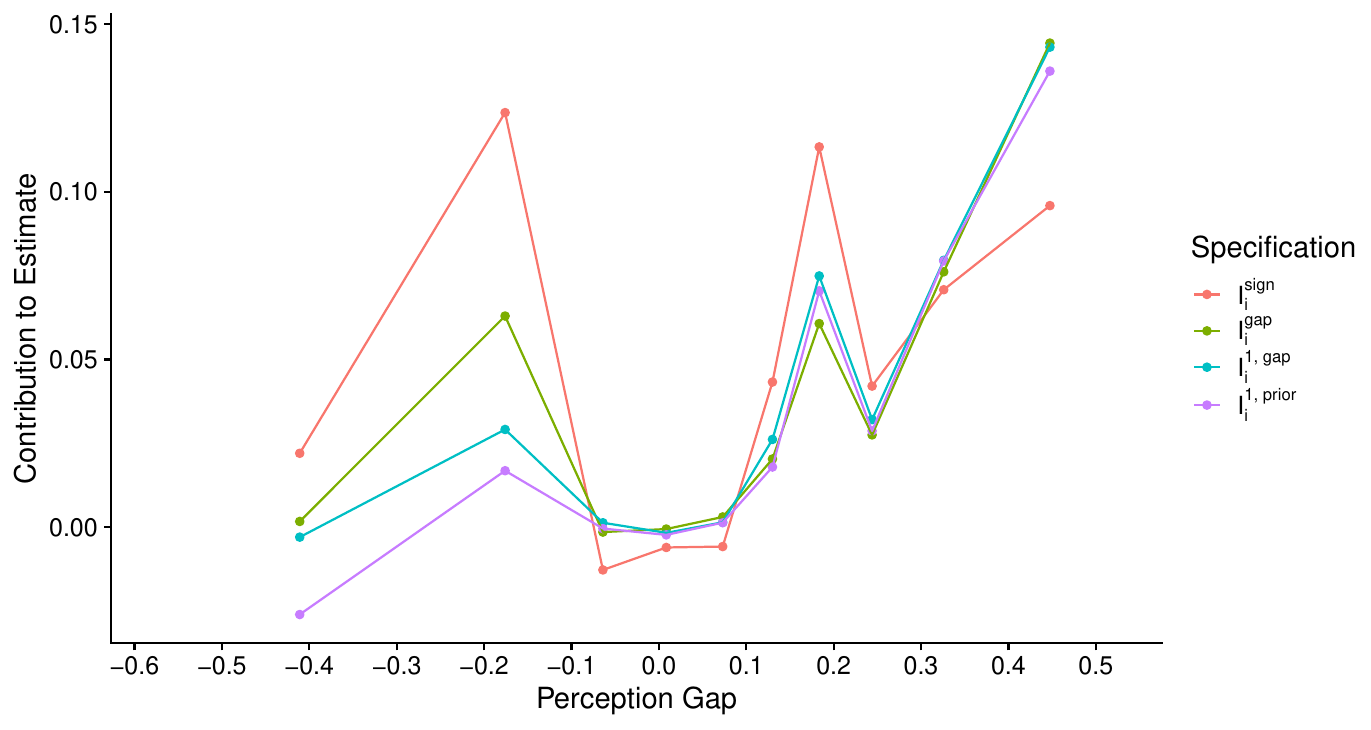}
    \caption{Perception Gap Characterization of TSLS---Data from \cite{jager2022worker}}
    \label{fig:jager_weight_Y}
    \floatfoot{\textit{Note:} Each point represents the characterized TSLS estimate for the corresponding decile bin, using intended negotiation probability as the outcome. For example, the leftmost point corresponds to the characterized estimate for those in the lowest decile bin, (-0.572,-0.228]. For a description of the characterization procedure, see Appendix \ref{main:sec:complier.characterization}.}
\end{figure}

\end{document}